\documentclass[11pt]{iopart}

\usepackage{iopams}  
 
\expandafter\let\csname equation*\endcsname\relax

\expandafter\let\csname endequation*\endcsname\relax

\usepackage{amsmath}

\usepackage{epsfig,amssymb,mathrsfs,amsthm,graphicx,float,color}
\usepackage[matrix,frame,arrow]{xypic}
\usepackage[matrix,frame,arrow]{xy}
\usepackage{amsmath}

\newcommand{\qw}[1][-1]{\ar @{-} [0,#1]}

\newcommand{\gate}[1]{*{\xy *+<.6em>{#1};p\save+LU;+RU **\dir{-}\restore\save+RU;+RD **\dir{-}\restore\save+RD;+LD **\dir{-}\restore\POS+LD;+LU **\dir{-}\endxy} \qw}

\newcommand{\measureD}[1]{*{\xy*+=+<.5em>{\vphantom{\rule{0em}{.1em}#1}}*\cir{r_l};p\save*!R{#1} \restore\save+UC;+UC-<.5em,0em>*!R{\hphantom{#1}}+L **\dir{-} \restore\save+DC;+DC-<.5em,0em>*!R{\hphantom{#1}}+L **\dir{-} \restore\POS+UC-<.5em,0em>*!R{\hphantom{#1}}+L;+DC-<.5em,0em>*!R{\hphantom{#1}}+L **\dir{-} \endxy} \qw}

\newcommand{\multimeasureD}[2]{*+<1em,.9em>{\hphantom{#2}}\save[0,0].[#1,0];p\save !C *{#2},p+LU+<0em,0em>;+RU+<-.8em,0em> **\dir{-}\restore\save +LD;+LU **\dir{-}\restore\save +LD;+RD-<.8em,0em> **\dir{-} \restore\save +RD+<0em,.8em>;+RU-<0em,.8em> **\dir{-} \restore \POS !UR*!UR{\cir<.9em>{r_d}};!DR*!DR{\cir<.9em>{d_l}}\restore \qw}

\newcommand{\multigate}[2]{*+<1em,.9em>{\hphantom{#2}} \qw \POS[0,0].[#1,0];p !C *{#2},p \save+LU;+RU **\dir{-}\restore\save+RU;+RD **\dir{-}\restore\save+RD;+LD **\dir{-}\restore\save+LD;+LU **\dir{-}\restore}
\newcommand{\ghost}[1]{*+<1em,.9em>{\hphantom{#1}} \qw}

\newcommand{\Qcircuit}[1][0em]{\xymatrix @*[o] @*=<#1>}  
 \renewcommand{\Qcircuit}[1][0em]{\xymatrix @*=<#1>}

\newcommand{\pureghost}[1]{*+<1em,.9em>{\hphantom{#1}}}
\newcommand{\multiprepareC}[2]{*+<1em,.9em>{\hphantom{#2}}\save[0,0].[#1,0];p\save !C
  *{#2},p+RU+<0em,0em>;+LU+<+.8em,0em> **\dir{-}\restore\save +RD;+RU **\dir{-}\restore\save
  +RD;+LD+<.8em,0em> **\dir{-} \restore\save +LD+<0em,.8em>;+LU-<0em,.8em> **\dir{-} \restore \POS
  !UL*!UL{\cir<.9em>{u_r}};!DL*!DL{\cir<.9em>{l_u}}\restore}
\newcommand{\prepareC}[1]{*{\xy*+=+<.5em>{\vphantom{#1\rule{0em}{.1em}}}*\cir{l^r};p\save*!L{#1} \restore\save+UC;+UC+<.5em,0em>*!L{\hphantom{#1}}+R **\dir{-} \restore\save+DC;+DC+<.5em,0em>*!L{\hphantom{#1}}+R **\dir{-} \restore\POS+UC+<.5em,0em>*!L{\hphantom{#1}}+R;+DC+<.5em,0em>*!L{\hphantom{#1}}+R **\dir{-} \endxy}}
\newcommand{\poloFantasmaCn}[1]{{{}^{#1}_{\phantom{#1}}}}

\usepackage{subfigure}
\newcommand{\R}{\mathbb{R}}
\newcommand{\C}{\mathbb{C}}

\newcommand{\set}[1]{\mathsf{#1}}
\newcommand{\grp}[1]{\mathsf{#1}}
\newcommand{\spc}[1]{\mathcal{#1}}

\def\d{{\rm d}}

\newcommand{\Lin}{\mathsf{Lin}}

\def\>{\rangle}
\def\<{\langle}
\def\kk{\>\!\>}
\def\bb{\<\!\<}
\newcommand{\st}[1]{\mathbf{#1}}

\newcommand{\map}[1]{\mathcal{#1}}
\newcommand{\Herm}{\mathsf{Herm}}

\newcommand{\St}{{\mathsf{St}}}

\newcommand{\Comb}{{\mathsf{Comb}}}

\newcommand{\Dual}{{\mathsf{Dual}}}

\newtheorem{theo}{Theorem}

\newtheorem{lemma}{Lemma}
\newtheorem{prop}{Proposition}
\newtheorem{cor}{Corollary}
\newtheorem{defi}{Definition}

\newtheorem{eg}{Example}
 
\newcommand{\Proof}{{\bf Proof. \,}}

\begin{document}
\title{%
Optimal quantum  networks and one-shot entropies
}
\author{Giulio Chiribella} 
\address{Department of Computer Science, The University of Hong Kong, Pukfulam Road, Hong Kong}
\address{Canadian Institute for Advanced Research,
CIFAR Program in Quantum Information Science, Toronto, ON M5G 1Z8}
\author{Daniel Ebler} 
\address{Department of Computer Science, The University of Hong Kong, Pukfulam Road, Hong Kong}
\begin{abstract}
We develop a semidefinite programming method for the optimization of quantum networks, including both causal networks and networks with  indefinite causal structure.   Our method  applies to  a broad class of  performance measures, defined operationally in terms of interactive tests set up by a verifier. We show that the optimal performance  is equal to a max relative entropy, which quantifies the informativeness of the test.   Building on this result,  we extend the notion of conditional min-entropy from quantum states to quantum causal networks.   The optimization method is illustrated in a number of applications, including the  inversion, charge conjugation, and controlization of an unknown unitary dynamics.    In the non-causal setting, we show a proof-of-principle application to the maximization of the winning probability in a non-causal quantum game.   
\end{abstract}
\maketitle

\section{Introduction}


Advances in quantum communication \cite{aspelmeyer2005advances,pirandola2015advances,wang-2015-nature}  and in the integration of quantum hardware \cite{politi2009integrated,crespi2011integrated,barends2013coherent,devoret2013superconducting,zwanenburg2013silicon} are pushing towards  the realization of networked quantum information systems, such as  quantum communication networks \cite{cirac1997quantum,acin2007entanglement,kimble-2008-nature,perseguers2010quantum,yin-2013-oe}    and 
  distributed   quantum computing  \cite{broadbent-2009-book,beals-2013-prs,barz-2012-nature}. 
       Networks of interacting quantum devices   are attracting  interest also at the theoretical level, providing a framework for quantum games \cite{gutoski2007toward} and protocols \cite{chiribella-dariano-2009-pra,chiribella2013short,portmann2015causal}, insights on the foundations of quantum mechanics \cite{chiribella-dariano-2009-pra,chiribella2011informational,hardy2011reformulating,hardy2012operator},  a starting point for  a general theory of Bayesian inference \cite{tucci1995quantum,leifer2008quantum,leifer2013towards,henson2014theory,ried2015quantum,chaves2015information,pienaar2015graph,costa2015quantum} and for the development of  models  of  higher-order quantum computation \cite{selinger2004towards,chiribella-2009-arxiv,chiribella-dariano-2013-pra}.

The network scenario motivates a new set of optimization problems,  where the goal is not to optimize individual devices, but rather to optimize how different devices interact with one another.     In many situations, the devices operate in a well-defined causal order---this is the case, for example, in the circuit model of quantum computing, where computations are implemented by   sequences of gates \cite{nielsen2000quantum,kitaev2002classical}.  Recently,   researchers have started to investigate more general situations,  where the causal order can be  in a  quantum superposition  \cite{chiribella-2009-arxiv,chiribella-2012-pra,chiribella-dariano-2013-pra,colnagh-2012-pla,araujo-2014-prl,portmann2015causal} or can  be indefinite  in other more  exotic ways, in principle  compatible with quantum mechanics    \cite{hardy-2009-book,oreshkov-2012-nature,chiribella-dariano-2013-pra,baumeler-2013-arxiv,morimae2014acausal,baumeler-wolf-2016,akibue2016entanglement}.  
 In these new situations, optimizing quantum networks is important,  for at least three reasons: First,  in order to establish an advantage, one has to first find the  optimal performances achievable in  a  definite  causal order.  Second, finding the maximum advantage requires an optimization over all  non-causal networks. This is an essential step for assessing the power of the new, non-causal models of information processing.  Third, identifying the ultimate performances  achieved in the absence of pre-defined causal structure is expected to shed light on the interplay between quantum mechanics and spacetime.   

In this paper we develop a semidefinite programming  (SDP) approach to the optimization of quantum networks.
  We start  by analyzing scenarios with  definite causal order, choosing an operational measure of performance,  quantified by  how much the network scores in a given test.   The test consists in   sending  inputs to the devices, performing local computations, and finally measuring the outputs.  Tests of this type are also important in the theory of  quantum interactive proof systems \cite{kitaev-watrous-2000}, wherein they are used to model the interaction between a prover and a verifier.   The  input-output behavior of  a quantum causal network  is described in  the framework of \emph{quantum combs} \cite{chiribella-dariano-2008-prl,chiribella-dariano-2009-pra}  (also known as \emph{quantum strategies}  \cite{gutoski2007toward}), which associates a positive operator to any given sequence of quantum operations.  In this framework, the optimization is a semidefinite program. We work out the dual optimization problem, showing that   the  maximum  score  is quantified by a one-shot entropic quantity that characterizes the  informativeness  of the test.  
This quantity extends to  networks   the notion of \emph{max relative entropy}   \cite{renner-wolf-2004-book,renner-2008-thesis,datta-renner-2009,
konig-renner-2009-ieee}  (see also  the monograph \cite{tomamichel-2015-book}).    Building on the connection with the max relative entropy, we define a measure of  the amount of correlations that a causal network can generate over time.  This quantity is based on the notion of {\em conditional min-entropy} \cite{datta-renner-2009,konig-renner-2009-ieee}, originally defined for quantum states and extended here to quantum causal networks.

After discussing the causal case, we turn our attention to  quantum networks with indefinite causal order. Some of these networks arise when  multiple quantum devices  are connected in a way that is controlled by the state of a quantum system  \cite{chiribella-dariano-2009-pra,colnagh-2012-pla,portmann2015causal}.   Some other networks are not built  by  linking up   individual devices \cite{oreshkov-2012-nature}. They are ``networks" in a generalized sense: they are spatially distributed objects that can interact with a set of local devices.   The description of these generalized networks is trickier, because we cannot specify their behaviour in terms of the behaviour of individual quantum devices. Instead, we must characterize them through     the way they respond to external inputs.      More specifically, a general quantum network is specified by a map that accepts as input the operations taking place in local laboratories  and returns as output an operation, as Figure \ref{causalladder}.     Maps that transform quantum operations  are known as \emph{quantum supermaps}.   They were originally introduced in the causal scenario    \cite{chiribella-dariano-2008-epl,chiribella-dariano-2009-pra}   and later generalized to the case of networks with indefinite causal structure \cite{chiribella-dariano-2009-pra,oreshkov-2012-nature,chiribella-dariano-2013-pra}.     These maps can be represented by positive operators, subject to a set of constraints that guarantee that valid operations are transformed into valid operations.  Again, the form of these constraints leads to semidefinite programs. In this case, we find that the  maximum score    can be expressed in terms of a  max  relative entropy, here named  the \emph{max relative entropy of signalling}, which quantifies  the deviation from the set of no-signalling channels.   In addition, we characterize the max relative entropy between two non-causal network, showing that it is equal to the maximum of the  max relative entropy over all the states that can be generated by interacting with the two networks.   This result opens the way to the definition of hypothesis testing protocols to probe the fundamental structure of spacetime, by testing the possibility of exotic non-causal networks against the null hypothesis that events have a well-defined causal structure.


\begin{figure}
\begin{center}
\hspace*{0cm}\includegraphics[scale=0.5]{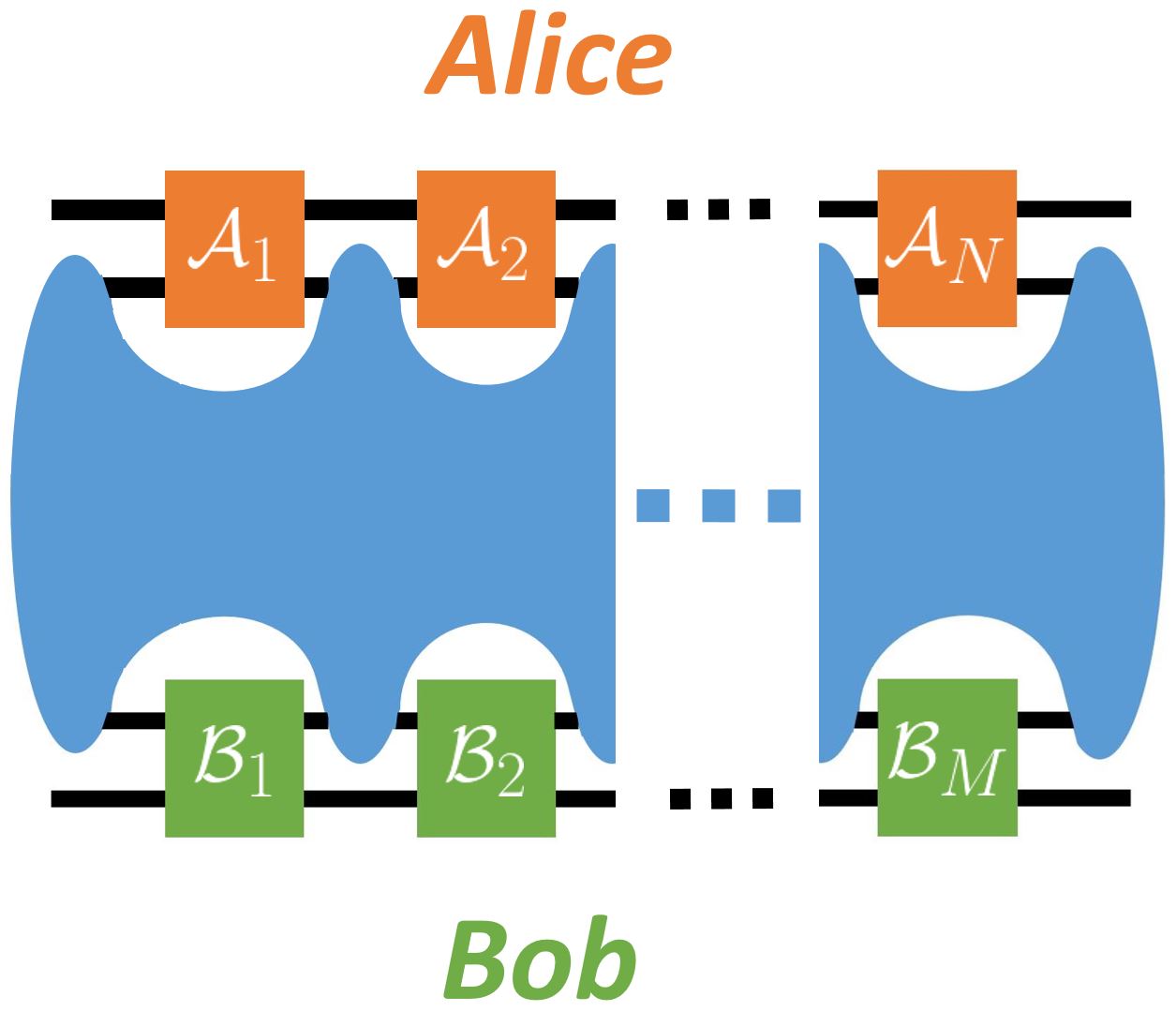}
\end{center}
\caption{Generalized network (in blue) interacting with two sequences of local devices in Alice's laboratory (orange boxes) and in Bob's laboratory (green boxes).    Devices acting in the same laboratory are applied in a well-defined causal order, corresponding to the direction from left to right in the picture. However, no causal order is assumed between the devices  in the two laboratories.  }
\label{causalladder}
\end{figure}


  To illustrate the general method, we  provide a number  of applications to concrete tasks, involving the optimization of both causal and non-causal networks.  For the optimization in the causal setting, we consider the tasks of    inverting an unknown unitary dynamics,  simulating the evolution of a charge conjugate particle, and adding control to an unknown unitary gate.     Looking at these tasks in terms of network optimization is a relatively new approach and here we provide the first optimized solutions.   
 For the optimization in the non-causal setting, we illustrate our method by analyzing the  non-causal game introduced by Oreshkov, Costa, and Brukner \cite{oreshkov-2012-nature}.  In this case, we fix  the operations performed by the players  (as in \cite{oreshkov-2012-nature}) and we search for the non-causal network that offers the largest advantage for these operations.   Using the SDP  approach, we obtain  a simple   proof   of the optimality of the network presented in  Ref.  \cite{oreshkov-2012-nature}.     Optimality  can also be derived from a recent result of  Brukner \cite{brukner-2015-njp}, who considered a more general scenario where  the players' operations are not fixed, but rather subject to optimization.    When the operations are fixed as in   Ref.  \cite{oreshkov-2012-nature}, however, our  SDP  technique yields  a significantly shorter   optimality proof.  The simplification in this restricted scenario suggests that   SDP  may prove useful  also for the broader scope  of identifying a non-causal  analogue of the Tsirelson bound, which was the motivating problem of Ref. \cite{brukner-2015-njp}.

The paper is organized as follows.  In Section \ref{preliminaries} we introduce the  framework of quantum combs and the characterization of quantum causal networks.    In Section \ref{sec:semidefinite} we  review the basic facts about semidefinite programming and establish a general relation with  the max relative entropy.  The  general result is applied to  quantum causal networks  in Section \ref{convexopt} and is  then used  to define a suitable extension of the conditional min-entropy  (Section \ref{sec:conditional}) and of the max relative entropy (Section \ref{sec:maxcausal}).   
In Sections \ref{non-causal} and \ref{sec:maxnoncausal} we extend  the results to quantum networks without predefined causal structure.    Our  techniques are illustrated in Section \ref{applications}, where we present applications to the tasks of inverting unknown evolutions, simulating charge conjugation, controlling unitary gates, and  maximizing the winning probability in a non-causal quantum game.  
Finally, the conclusions are drawn in Section \ref{sec:conclu}.

\section{The framework of quantum combs \label{preliminaries}}

In this section we introduce the concepts required for the optimization  of quantum causal networks. First of all, we review the connection between quantum channels and   operators. Then, we present the basics of the framework of quantum combs. 
\subsection{Quantum operations, quantum channels, and the Choi isomorphism}
Quantum operations \cite{kraus-1983} describe the most general transformations of quantum systems, including both the reversible transformations associated to unitary gates and the irreversible transformations due to measurements.    A quantum operation with input system $A$ and output system $B$  is a completely positive trace non-increasing map  $\map C$, transforming operators    on the input Hilbert space $\spc H_A$ into operators on the output Hilbert space $\spc H_B$.  We will often use the diagrammatic notation 
\begin{align}\label{C}
\Qcircuit @C=1em @R=0.2em @!R
{    &  \qw \poloFantasmaCn{A}  & \gate {  \map C}  & \qw \poloFantasmaCn{B}  &\qw    }   \, .
\end{align}
 We say that the  quantum operation $\map C$ in the above diagram is \emph{of type $A\to B$}.

 When  system $A$ is trivial---that is,~when its Hilbert space is one-dimensional---the quantum operation $\map C$ corresponds to the preparation of a \emph{state} of system $B$, diagrammatically represented as 
$\Qcircuit @C=1em @R=0.2em @!R
{      \prepareC{\rho}    &  \qw \poloFantasmaCn{B}  &  \qw      &  }$. 
When system $B$ is trivial, the quantum operation  $\map C$ in Eq. (\ref{C}) corresponds to a measurement \emph{effect} on system $A$ and is represented as 
$\Qcircuit @C=1em @R=0.2em @!R
{     &  \qw \poloFantasmaCn{A}          & \measureD{P}     }$.  
  Measurement effects  are positive (semidefinite) operators $P$ satisfying   $P\le I$, where $I$ is the identity operator on the system's Hilbert space.      Effects are associated to the outcomes of measurements and the probability of the outcome corresponding to the effect $P$ is given by the Born rule 
\begin{align} 
\Qcircuit @C=1em @R=0.2em @!R
{      \prepareC{\rho}    &  \qw \poloFantasmaCn{A}  &  \measureD{P}      }  =  \Tr [  P  \rho  ] \, ,
\end{align}
where $\rho$ is the state of the system before the measurement.   In the special case where $P$ is the identity  operator,   we represent the corresponding effect as   $\Qcircuit @C=1em @R=0.2em @!R
{     &  \qw \poloFantasmaCn{A}          & \measureD{\Tr}     }$ .  

In general,  quantum measurement processes are described by  \emph{quantum instruments}.  A quantum instrument  with input $A$ and output $B$ is a collection of quantum operations $\{ \map C_x\}_{x\in\set X}$ of type $A\to B$, subject to the condition that the sum $\sum_{x\in\set X}  \map C_x$  is trace-preserving.    Each quantum operation corresponds to one alternative outcome $x$ and the probability that the quantum operation $\map C_x$ takes places on a given input state $\rho$ is given by  
 \begin{align} 
\Qcircuit @C=1em @R=0.2em @!R
{      \prepareC{\rho}    &  \qw \poloFantasmaCn{A}  &    \gate{\map C_x}  &  \qw \poloFantasmaCn{B}  &  \measureD{\Tr}      }    ~=~  \Tr [\,  \map C_x  ( \rho)  \,] \, ,
\end{align}
When the instrument $\{\map C_x\}_{x\in\set X}$ has a single outcome,  say $x_0$,  the corresponding process is \emph{deterministic}, meaning that one can predict in advance that the outcome will be $x_0$.    In this case, the quantum operation $\map C_{x_0}$ is trace preserving.   
Trace preserving quantum operations  are also known as \emph{quantum channels}.

Completely positive maps can be represented by positive  operators.   Let $\Lin (\spc H)$ be the space of linear operators on the Hilbert space $ \spc H$ and    let  $\mathcal{C}$ be a completely positive map    transforming operators in $\Lin  (\mathcal{H}_0)$ into operators on $   \Lin  (\mathcal{H}_1)$.    Then,  the map $\map C$ can be represented by a positive operator $C \in \Lin  (\spc H_1 \otimes \spc H_0)$, defined as
\begin{align}\label{choi}
{C}
=(\mathcal{C} \otimes   \map I_0)(  \,  |   I  \kk  \bb I  |  \,  )
\end{align} 
where   $\map I_0$ denotes the identity map on  $\Lin (\spc H_0)$ and  $|I  \kk$  is the unnormalized maximally entangled state $|  I\kk  = \sum_i |i\rangle  |i\rangle    \in  \spc H_0\otimes \spc H_0  $. The operator $C$ is known as  the \emph{Choi operator}    \cite{choi-1975}.  

Quantum operations and quantum channels can be characterized in terms of their Choi operators:  a positive operator  $Q  \in  \Lin(\spc H_1\otimes \spc H_0)$ corresponds to a quantum operation if and only if it satisfies the condition 
\begin{align} \Tr_{1}  [   Q]  \le   I_0  \, ,\end{align}
where $\Tr_1$ denotes the partial trace over the Hilbert space $\spc H_1$,  $I_0$ denotes the identity operator on the Hilbert space $\spc H_0$, and $\le$ denotes the  standard operator order: $A  \le B$  iff $  \<\varphi | A  |\varphi\> \le \<\varphi|  B  |\varphi\>$, $\forall  |\varphi\>\in \spc H_0$.    
A positive operator  $C  \in  \Lin(\spc H_1\otimes \spc H_0)$ corresponds to a quantum channel if and only if it satisfies the condition 
\begin{align}\label{choichannel}  \Tr_{1}  [   C]  =   I_0  \, . \end{align}

\subsection{The link product}

Two quantum operations can be connected with each other, as long as the output of the first operation matches the input of the second.    
      At the level of Choi operators, the connection is implemented by the operation of  \emph{link product}  \cite{chiribella-dariano-2008-prl}, denoted as $*$.     To define the link product, it is convenient to introduct the following  shorthand notation:  if    $A$  is an operator  on $\mathcal{H}_{X} \otimes \spc H_{Y}$  and $B$ is an operator on $\spc{H}_{Y} \otimes  \spc  H_{ Z}$, then we use the notation $AB$ for the product
\begin{align}
AB:=(A\otimes I_{Z})(I_{X}\otimes B) \ .
\end{align}
With this notation, the link product of $A$ and $B$ is the operator $A*B$ defined as 
\begin{align}\label{link}
A* B:=\Tr_Y  \left [A \, B^{T_Y}\right]
\end{align}
where  $B^{T_Y}$ denotes the partial transpose of $B$ with respect to the Hilbert space  $\spc H_Y$.  Note that the definition of the link product presupposes that the Hilbert spaces have been labelled: in order to compute the link product, one needs to take the partial transpose and the trace on the Hilbert space in common between $A$ and $B$.       
  Mathematically, the  partial transpose in the r.h.s. of  Eq. (\ref{link}) is essential to guarantee that the link product of  two positive operators is a positive operator   \cite{chiribella-dariano-2008-prl}.   As a counterexample, think of  the case where      $\spc H_X$, $\spc H_Y$, and $\spc H_Z$ are two-dimensional  and $A$ and $B$ are projectors on a maximally entangled state: in this case,  removing the partial transpose results in a non-positive $A*B$). Physically, the role of the partial transpose can be understood in terms of entanglement swapping \cite{yurke1992,zukowski1993}. Thanks to the partial transpose, the link product can be expressed as 
  \begin{align}
  A*B    =    \Tr_{Y} \Tr_{Y'}  [   (  A_{X  Y}   \otimes B_{  Y'  Z })  \,  (  I_X \otimes | I\kk\bb I  |  \otimes I_Z ) ] \, ,
  \end{align}  
  where   $|I\kk  :=  \sum_{n=1}^{  d_Y}  \,  |n\>|n\>$ is the unnormalized maximally entangled state on $ \spc   H_Y\otimes \spc H_{Y'} $,  $\spc H_{Y'}$ being an identical copy of $\spc H_Y$.     This means that,   up to normalization, the link product $A*B$ is the state obtained when a Bell measurement, performed on the states $A/\Tr A$ and $B/\Tr B$, yields the outcome corresponding to the projector $ | I\kk\bb I | /d_Y$.     At the fundamental level, the possibility of representing operations as states and their composition as entanglement swapping follows from the Purification Principle---the property that every state can be obtained as the marginal of a pure state, unique up to reversible transformations \cite{chiribella2010probabilistic}.

The link product is associative, namely 
\[A  * (  B*C)   =     (A   * B )   *C  \, ,\]
for all operators $A$, $B$, and $C$. Moreover, the link product  is commutative, up to re-ordering of the Hilbert spaces: in formula, 
\[A* B \simeq B* A \, ,\]
having used  the notation $A  *B  \simeq  B*A  $ to mean  $A *  B  =  {\tt SWAP}_{XZ}     \,  (\,B*A\,)  \,   {\tt SWAP}_{XZ}$,  where  $ {\tt SWAP}_{XZ}   $ is the unitary operator that swaps the spaces $\spc H_X$ and $\spc H_Z$.    From now on we will omit the swaps, implicitly understanding that the Hilbert spaces have been   reordered in the right way wherever needed.      

Using the above notation, we have the following 
\begin{prop}[\cite{chiribella-dariano-2008-prl}]\label{prop:link}
Let $\map A$    be a quantum operation transforming operators on $\spc H_0$ to operators on $\spc H_1$, let $\map B$ be a quantum operation  transforming operators on $\spc H_1$ to operators on $\spc H_2$, and let $\map C =  \map B \map A$ be the quantum operation resulting from the composition of $\map A$ and $\map B$. Then, one has 
\[  C   :  =    A  *  B  \, ,\]
where $A$, $B$, and $C$ are the Choi operators of $A$, $B$, and $C$, respectively.  
\end{prop}

   In the next paragraph we will use the link product to construct the Choi operator of quantum networks consisting of multiple interconnected quantum operations.

\subsection{Quantum causal networks and quantum combs} 
A quantum network is a collection of quantum devices connected with each other. We will call the network \emph{causal} if there are no loops connecting the output of a device to the output of the same device.  Mathematically, a quantum causal network can be represented by a direct acyclic graph,  where each vertex of the graph corresponds to a quantum device---cf. Figure \ref{DAG}. For every DAG, one can always define a total ordering of the vertices, through a procedure known as  topological sorting \cite{thomas2001introduction}.   Using this fact, one can always represent the a quantum causal network as an ordered  sequence of quantum devices, such as 
\begin{align}\label{comb}
\Qcircuit @C=1em @R=0.2em @!R
{     & \qw \poloFantasmaCn{A_{1}^{\rm in}}  &  \multigate{1}{\map C_1}   &  \qw \poloFantasmaCn{A^{\rm out}_1}    &\qw &    & \qw\poloFantasmaCn{A_2^{\rm in}}  &  \multigate{1}{\map C_2}   & \qw \poloFantasmaCn{A_2^{\rm out}}   &  \qw  &\dots  &&\qw   \poloFantasmaCn{A_N^{\rm in}} &  \multigate{1}{\map C_N}  &  \qw  \poloFantasmaCn{A_N^{\rm out}}   &\qw \\
     & & \pureghost{\map C_1}  &   \qw &  \qw &\qw &\qw & \ghost{\map C_2}  &\qw  &  \qw &  \dots   &&  \qw & \ghost{\map C_N}    & &    }
\quad ,
\end{align} 
where $A_j^{\rm in}$   ($A_j^{\rm out}$) denotes the input  (output) system of the network at the $j$-th time step. 

\begin{figure}
\begin{center}
\hspace*{0cm}\includegraphics[scale=0.3]{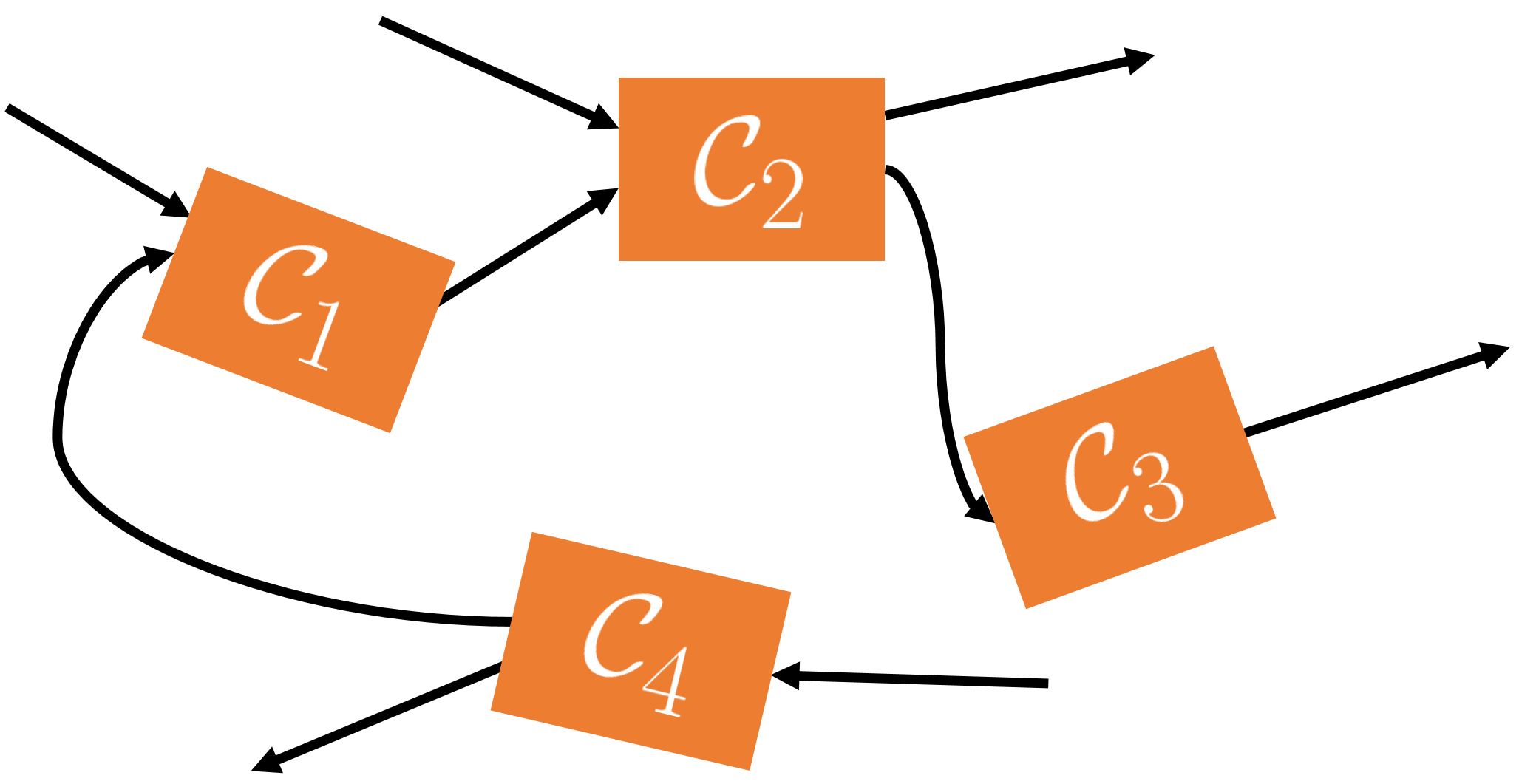}
\end{center}
\caption{A quantum causal network is a directed acyclic graph, whose nodes (orange boxes in the picture) represent  quantum devices and whose directed edges indicate the input/output direction. }
\label{DAG}
\end{figure}

We say that a network is \emph{deterministic} if all devices in the network are deterministic,~i.~e.~if they are described by quantum channels.  Using the link product, we associate a Choi operator to the network:   specifically, if the individual channels in the network have Choi operators $C_1,\  C_2 \, ,  \dots \, ,  C_N$, then the network has Choi operator 
\begin{align}
C=C_1 * C_2 * C_3 * \cdots * C_N \ .
\end{align}

The Choi operator of a deterministic  network is called a \emph{quantum comb}  \cite{chiribella-dariano-2008-prl,chiribella-dariano-2009-pra}, or also  a \emph{quantum strategy} \cite{gutoski2007toward}.    The quantum comb $C$ is a positive operator on $\bigotimes_{ j=1}^{N}  \,   \left(  \spc H_j^{\rm out} \otimes \spc H_j^{\rm in}\right) $, where $\spc H_j^{\rm in}$   ($\spc H_j^{\rm out}$)  is the Hilbert space of system $A_j^{\rm in}$  ($A_j^{\rm out}$).  
 Quantum combs can be characterized as follows:  
 \begin{prop}[\cite{chiribella-dariano-2008-prl,chiribella-dariano-2009-pra,gutoski2007toward},]    A positive operator  $C$   is a  quantum comb  if and only if it  satisfies  the linear constraints
  \begin{align}
\Tr_{A_n^{\rm out}}[C^{(n)}]=I_{A_n^{\rm in}} \otimes C^{(n-1)}  \qquad  & \forall n\in \{1,\ldots ,N\}   
  \label{tr}  \, ,
\end{align}
where      $\Tr_A$ is the partial trace over the Hilbert space  $\spc H_{A}$,  $C^{(n)}$  is a suitable operator on $  \spc H_n   :  = \bigotimes_{ j=1}^{n}  \,   \left(  \spc H_j^{\rm out} \otimes \spc H_j^{\rm in}\right) $,  $C^{(N)}   : =   C$, and $C^{(0)}:  = 1$. 
\end{prop} The constraints in Eq.~(\ref{tr}) are a direct consequence of the normalization condition of quantum channels, expressed by Eq. (\ref{choichannel}). Physically, the positive operator $C^{(n)}$ represents the subnetwork transforming the first  $n$ inputs to the first $n$ outputs.  

 We   denote by 
 \[\set{Comb}   \left(    A^{\rm in}_1  \to      A^{\rm out}_1 \, ,  A^{\rm in}_2  \to       A^{\rm out}_2 \,  , \dots  \, ,A^{\rm in}_N  \to    A^{\rm out}_N    \right) \] the set of positive semidefinite operators satisfying the constraint (\ref{tr}).  
  When there is no ambiguity,  we will simply write $\set {Comb}$.


\subsection{Quantum testers and the generalized Born rule}

   So far we considered deterministic networks, resulting from the connection of quantum channels.  
However, it is also useful to consider   networks containing measurement devices, which may generate random outcomes.   We call such networks \emph{non-deterministic}. Non-deterministic  quantum networks can be thought as the quantum version of classical electric networks containing measurement devices, such as voltmeters and ammeters.  Like these classical relatives are useful for testing the behaviour of electrical circuits, quantum non-deterministic networks  are useful for testing the behaviour of  quantum circuits, or, slightly more broadly,  physical processes consisting of multiple time steps.  

    An example of non-deterministic network is  the following
        \begin{align}\label{tester} 
   \Qcircuit @C=1em @R=0.2em @!R
{    & \pureghost{\rho}  &   \qw  &\qw &\qw &\qw & \ghost{\map D_1}  &  \qw &\qw &  \dots   &&  \qw  & \ghost{\map D_{N-1}}  &\qw   & \qw   &  \qw  &    \qw & \ghost{\{P_x\}_{x\in\set X}}\\
 & \pureghost{\rho}  &     & & & & \pureghost{\map D_1}  &   & &    &&    & \pureghost{\map D_{N-1}}   & &    &   &     & \pureghost{\{P_x\}_{x\in\set X}}\\
      &  \multiprepareC{-2}{\rho} &   \qw \poloFantasmaCn{A_{1}^{\rm in}}   & \qw     &     & \qw \poloFantasmaCn{A_{1}^{\rm out}}  &   \multigate{-2}{\map D_1}   &  \qw \poloFantasmaCn{A_{2}^{\rm in}}     &    \qw  &\dots  &  &  \qw \poloFantasmaCn{A_{N-1}^{\rm out}}   &  \multigate{-2}{\map D_{N-1}}  &    \qw \poloFantasmaCn{A_{N}^{\rm in}}     & \qw &        & \qw \poloFantasmaCn{A_{N}^{\rm  out}}    &     \multimeasureD{-2}{\{P_x\}_{x\in\set X}}   } 
\quad ,
 \end{align}
  where $\rho$ is a quantum state, $(\map D_1,  \dots,  \map D_{N-1})$ is a sequence of quantum channels, and $\{P_x\}_{x\in\set X}$ is a positive operator-valued measure (POVM), describing a quantum measurement  on the last output system.  Networks of the type  (\ref{tester})  can be used to probe quantum networks of the type (\ref{comb}), as follows
 \begin{align}\label{connected}
\Qcircuit @C=1em @R=0.2em @!R
{ &    \multiprepareC{2}{\rho}  & \qw  & \qw & \qw & \multigate{2}{\map D_1}  &\qw &  \qw &  \dots    &   & &  \qw & \multigate{2}{\map D_{N-1}}  & \qw &   \qw  & \qw &  \multimeasureD{2}{\{P_x\}_{x\in\set X}}\\    
&    \pureghost{\rho}  &   &  &  & \pureghost{\map D_1}  & &  &  \dots    &   & &   & \pureghost{\map D_{N-1}}  &  &     &  &  \pureghost{\{P_x\}_{x\in\set X}}\\        
   &       \pureghost{\rho}    &     \qw \poloFantasmaCn{A_{1}^{\rm in}}      &   \multigate{1}{\map C_1}   &   \qw \poloFantasmaCn{A_{1}^{\rm out}}    &    \ghost{\map D_1}    &    \qw \poloFantasmaCn{A_{2}^{\rm in}}     &\qw  & \dots  & &&     \qw \poloFantasmaCn{A_{N-1}^{\rm out}}      &  \ghost{\map D_{N-1}}  &    \qw \poloFantasmaCn{A_{N}^{\rm in}}       &     \multigate{1}{\map C_{N}}  &    \qw \poloFantasmaCn{A_{N}^{\rm out}}    &   \ghost{\{ P_x\}_{x\in\set X}}  \\
   & &  &   \pureghost{\map C_1}  &   \qw &  \qw &   \qw &  \qw &\dots    && &\qw & \qw & \qw & \ghost{\map C_N}    &   &}
    \quad .  \qquad \qquad 
\end{align}  
When the two networks are wired together, the final measurement    produces one of the outcomes in the set $\set X$.  
 Using Proposition \ref{prop:link}, the probability of the outcome $x $ is can be computed as
\begin{align}
\nonumber p_x  &   =   \rho   *   C_1   *  D_1 *C_2  *   D_2  * \cdots  *    D_{N-1}   *C_N*   P^T_x \\
\nonumber  &  =  \left(   ~ \rho  *  D_1  *  D_2  *  \cdots  *  D_{N-1}  *  P_x^{T}  \, \right) \,  *\,  \left(~  C_1*C_2 *  \cdots  *  C_N^{\phantom{\dag}} \right)  \\
\nonumber  & =   T_x  *  C  \\
\label{genborn}&  =  \Tr\left[ \,  T_x  \,  C^T \, \right]  \, ,
\end{align}
where  $C$ is the Choi operator of the tested network, $C^T$ is the transpose of $C$, and $ \{  T_x\}_{x\in\set X}$  is the collection of operators defined by  
\begin{align}\label{tix}   T_x  :  =  \rho *   D_1   *  D_2  *   \cdots  *  D_{N-1}  *  P^T_x  \end{align}
(here the transpose  of $P_x$ is needed because, according to Definition \ref{choi},  the Choi operator of  the quantum operation  $\map Q_x   (\cdot) =   \Tr  [ P_x  \cdot ]$ is $P_x^T$ instead of $P_x$).  

We call the set of operators $\st T  =   \{  T_x\}_{x\in\set X}$   a     \emph{quantum tester} and Eq.  (\ref{genborn})  the \emph{generalized Born rule} \cite{chiribella2008memory,chiribella-dariano-2009-pra,ziman2008process}.   The quantum tester  $\st T$  describes  the response of the non-deterministic network  (\ref{tester}) when connected to external devices.  Quantum testers are a useful abstraction  in many applications, such as quantum games \cite{gutoski2007toward} and cryptographic protocols \cite{chiribella2013short,portmann2015causal},    quantum interactive proof systems \cite{kitaev-watrous-2000},    quantum learning of gates \cite{bisio2010}, \cite{giovannetti2006quantum,giovannetti2011advances},  quantum channel discrimination   \cite{chiribella2008memory,gutoski2012,jenvcova2016conditions},   incompatibility of multitime quantum measurements \cite{sedlak2016incompatible},  tomography of quantum channels \cite{ziman2008process,bisio2009optimal},  non-Markovian processes \cite{pollock2015complete,modi2012unification}, and causal models \cite{ried2015quantum}.

  Quantum testers can be characterized as follows:  
\begin{prop}[\cite{chiribella2008memory}] 
Let ${\st T}  $ be a collection of  positive operators on   $  \bigotimes_{ j=1}^{N}  \,   \left(  \spc H_j^{\rm out} \otimes \spc H_j^{\rm in}\right)  $.       $\st T$ is a quantum tester if and only if 
\begin{align}
\sum_{x\in \set X}  \,  T_x  =& I_{A_N^{\rm out}}\otimes \Gamma^{(N)} \nonumber \\
\Tr_{A_n^{\rm in}} \left [\, \Gamma^{(n)}\, \right]=&I_{A_{n-1}^{\rm out}}  \otimes \Gamma^{(n-1)} \ , \quad n=2,\ldots,N  \nonumber \\
\Tr_{A_1^{\rm in}}   \left[  \,  \Gamma^{(1)} \,  \right]  &  =1  \, ,
 \label{trtester}
\end{align} 
where each  $\Gamma^{(n)}$, $n=  1  ,\dots,   N$ is a positive operator on  $\spc H_{n}^{\rm in}  \otimes \left[     \bigotimes_{ j=1}^{n-1}  \,   \left(  \spc H_j^{\rm out} \otimes \spc H_j^{\rm in}\right)\right]$. 
\end{prop}

\subsection{Assessing the performance of a  quantum network}\label{subsec:test}  

Suppose that we are given black box access to a quantum network, whose  internal functioning is unknown to us.   
  Our goal is to assess how well the network fares in a desired task, such as  solving a desired  computational problem \cite{Deutsch73},   estimating an unknown parameter \cite{giovannetti2006quantum,giovannetti2011advances},  emulating a sequence of gates  \cite{bisio2010,marvian2016}, or replicating the action of a desired gate \cite{chiribella2008optimal,dur2015deterministic,chiribella2015universal,mivcuda2016experimental,chiribella2016quantum}.  
  
For example,  suppose that a manufacturer provides us with a special-purpose computer, designed to implement  a  quantum  search algorithm.  How can we test the performance of our device?    Since the computer is claimed to find the location of an item in a list, a natural approach  is to place the item in a set of random positions and then to check whether the answer provided by the computer is correct.  A simple measure of performance is given  the number of inputs on which  the computer gives the right answer.   More generally, one can assign different scores  depending on the distance between the correct answer and  the output of the computer.     Let us consider this example in more detail, as a concrete illustration of what it means to test a quantum network.    Suppose that the computer attempts at reproducing Grover's algorithm  \cite{grover-1996}, by interacting with unitary gates $U_i  =  2| i\>\<i|  -  I$ that encode the position of an item  $i$ in a list of $K$ items.   
    A possible test, illustrated in Figure \ref{grover},  
 is as follows:  
   \begin{enumerate}
   \item Prepare a ``position register" in the  maximally mixed state $\rho  =  I/K$.
   \item  Upon receiving an input from the computer, apply the control-unitary gate $  W   =  \sum_{i=1}^K \,  |i\>\<i|  \otimes U_i$ to the position register and the input.
   \item Repeat the previous operation until the computer returns an output. In this way, the input provided by the computer is processed by a gate $  U_i$, with $i$ chosen at random.
   \item Compare the output with the actual position, by performing a joint measurement to the position register and the output register.  The measurement is described by the POVM $\{  P_x\}_{x=  -K}^K$  with operators given by  
   \[P_x=\sum_{i   =  \max \{0, -x\}}^{\min \{K,  K-x \}} |i\rangle \langle i| \otimes |i+x \rangle \langle i+x | \, . \]   In this way,  the measurement outcome returns the deviation $x$ from the correct position
   \item  If the deviation is $x$,  assign  score  $\omega_x   =  1-   |x|/K$.    
    \end{enumerate}

\begin{figure}
\begin{center}
\hspace*{0cm}\includegraphics[scale=0.6]{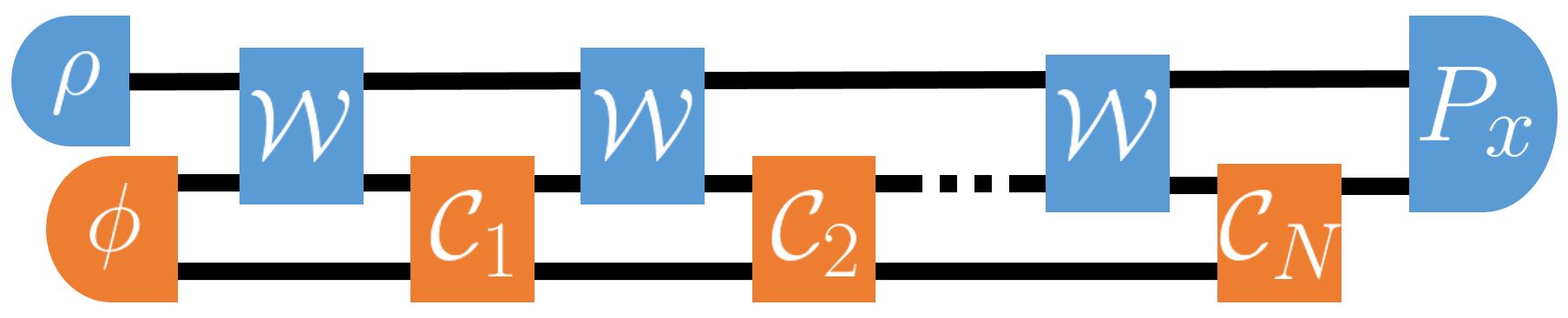}
\end{center}
\caption{A computer is designed to implement Grover's algorithm. The action of the computer  (in orange) is tested  by a   testing circuit (in blue), consisting in the preparation of randomly chosen input (encoded in the state  $\rho  =   I/K$   of a ``position register"), followed by the application of the control-unitary gate $W$, which, depending on the input, performs one of the unitaries $  U_i$.   In the end, the computer outputs an outcome $j$ (encoded in the output of the channel $\map C_N$), which is  compared with the position register  through a suitable quantum measurement (POVM $\{ P_x\}$), which  outputs the deviation $x =  i-j$.   } 
\label{grover}
\end{figure}
      
 Mathematically, the test is represented by the quantum tester  $\{  T_x\}_{x=  - K}^K$ with  
   \[   T_x  =   \rho    *       |W\kk\bb  W|   *  \cdots  * |W\kk\bb W|       *    P^T_x \, .\]
   The sequence of operations performed by the computer is  represented by the quantum comb 
   \[  C   =    |\phi\>\<\phi|  *   C_1  * \cdots  *C_{N}  \, ,\]   
and the probability of finding a deviation $x$ is given by
\[ p_x  =   T_x  *   C \,. \]   
The average score obtained by the computer can be expressed as 
   \begin{align*}
\omega    & =   \sum_{i,j=0}^K \,   \left(  1  -      \frac { |i-j|} K  \right)  \,     p_{i-j} \,\\
   &   =   \sum_{x= - K}^K   \,      \left(  1  -      \frac { |x|} K  \right)     \,    T_x*  C   \\
   &  =     \Omega  *  C   \, ,          
\end{align*}
where $\Omega$ is the operator   $\Omega :  =  \sum_x  \,    \left(  1  -      \frac { |x|} K  \right)   \,  T_x$.

Generalizing the above example, we   assess the performance of an unknown quantum network by referring to experiments where the unknown network is connected to  a ``testing network",  containing measuring devices.    The testing network will return an outcome  $x$, to which one can assign a ``score" $\omega_x$. In this way, the expected score serves as an operational measured of performance.   Specifically,  let $C$ be the quantum comb describing the tested network and let $\st T   =  \{  T_x \,  ,  x\in\set X\}$ be the quantum tester describing the testing network.  
Then, the average score is given by
\begin{align}
\nonumber \omega   &=    \sum_x  \,  \omega_x  \,    ( \,  T_x  *  C   \,) \\
&  =     \Omega  *  C   \qquad \Omega  :  =   \sum_x  \,      \omega_x\,  T_x    \, .
\end{align} 
Note that the performance  of the network $C$  is completely  determined  by  the operator $\Omega$, which we  call the \emph{performance operator}.  


For a given performance operator $\Omega$,  the maximum expected  score is given by  
 \begin{align} 
 \nonumber  \omega_{\max}     &:  =  \max_{C  \in  \Comb   }  \,          \Omega  *C  \\  
 &=    \max_{C  \in  \Comb  }  \,             \Tr  [\Omega \, C^T] 
\\
\nonumber  &=    \max_{C  \in  \Comb  }  \,             \Tr  [\Omega \, C] 
  \, .   
 \end{align}
 The third equality comes from the fact that the set of quantum combs is closed under transposition and, therefore, we can omit the transpose in Eq.  (\ref{genborn}).   Using the notation 
 \begin{align}\<A  ,  B\>  : =  \Tr [A  B] \, ,
 \end{align}
  we express the maximum score as 
 \begin{align}\label{omegamax}  \omega_{\max}     &:  =  \max_{C  \in  \Comb }  \,         \<   \Omega, C\> \, .
 \end{align}

The above equation shows that the search for the maximum score is a semidefinite program.   The basic tools needed to address it will be reviewed in the next section.

\section{Semidefinite programming and the max relative entropy}\label{sec:semidefinite}

\subsection{Basic facts about semidefinite programming}
Here we review the background  about semidefinite programming.   For further details, we refer the reader to  Watrous' lecture notes \cite{watrous-2011-notes}.  

 Let $\spc X$  and $\spc Y$ be two a Hilbert spaces and let $\Herm (\spc X)$ be the space of Hermitian operators on $\spc X$ and $\spc Y$, respectively.  
 
\begin{defi}A \emph{semidefinite program  (SDP)} is a triple $(\phi,A,B)$, where $A$ and $B$ are operators in  $ \Herm(\mathcal{X})$ and $  \Herm(\mathcal{Y})$, respectively, and  $\phi$ is a linear map from   $ \Herm(\mathcal{X})$ to $  \Herm(\mathcal{Y})$. 
\end{defi}

A semidefinite program is associated to an optimization problem in the  standard form
\begin{align}
\nonumber & {\rm maximize}  \  \langle A, X\rangle  \\
\nonumber &  {\rm subject~to }   \ \phi (X)=  B  \\
& X  \ge 0   \, . \label{primal1}
\end{align}
This problem is known as the \emph{primal}. The  \emph{dual} problem is 
\begin{align}
\nonumber & {\rm minimize}  \  \langle B, Y\rangle  \\
\nonumber &  {\rm subject~to }   \ \phi^\dag (Y)\ge   A    \\
&  Y  \in  \Herm   (\spc Y) \, , \label{dual1}
\end{align}
 where $\phi^\dag$ is the adjoint of $\phi$, namely the linear map defined by the relation 
 \[\< X  ,  \phi^\dag  (Y )  \>  =  \<  \phi (X),  Y \> \, ,  \qquad \forall  X\in  \Herm(\spc X) \, , \forall Y  \in \Herm(\spc Y) \, . \] 

The optimal values of the primal and dual problems, denoted as 
\[\omega_{\rm primal}   := \sup \ \langle A,X \rangle  \qquad {\rm and}  \qquad \omega_{\rm dual} := \inf \ \langle B,Y \rangle   \, ,
\]  
are related by duality:  For every semidefinite program, one has the weak duality  $\omega_{\rm primal} \leq \omega_{\rm dual}$.     The strong duality  $\omega_{\rm primal} = \omega_{\rm dual}$  holds under suitable conditions, provided by Slater's Theorem \cite{boyd2004convex}.  
 In this paper  we will use the following
\begin{prop}\label{Slater} Let $(\phi,A,B)$ be a semidefinite program.   If there exists a positive operator $X$  satisfying $\phi  (X)  =  B$  and an Hermitian operator $Y$ satisfying $\phi^\dag(Y)> A$, then $\omega_{\rm primal}=\omega_{\rm dual}$.
\end{prop}
For the proof, see  e.~g.~\cite{watrous-2011-notes}.

\subsection{The max relative entropy}

An important quantity in one-shot  quantum information theory is the \emph{max relative entropy}, introduced by Datta in Ref. \cite{datta-2009-ieee}:  
\begin{defi}  
Let  $A$ and $B$ be two positive operators on $\spc X$.  The max entropy of $A$ relative to $B$ is given by
 \begin{align}\label{dmax} D_{\max}  (A\,  \|   \,  B )  :   =   
   - \log   \max\{  w~|~  w  \,A \le B  \} \, ,
   \end{align}
   with the convention $\log  0  :  = -\infty$. 
   \end{defi}
The max relative entropy provides one way to quantify  the deviation of $A$ from $B$.  More generally, it is useful to  consider the deviation between $A$ and a \emph{set} of operators: 
\begin{defi}  
Let $A$ be a positive  operator on $\spc X$ and let $  \set S \subset \Herm  (\spc X)$ be a set of positive operators.   The  \emph{max  entropy} of   $A$   relative  the set $\set S$, denoted as $D_{\max}   (  A \,  \|   \,  \set S )$,    is the quantity defined by
\begin{align}
D_{\max}       (  A \,  \|    \,  \set S )    & : =     \inf_{  B  \in  \set S}  \,   D_{\max} (  A  \,  \|  \,    B   )   \, .       
\end{align}
\end{defi}
The max relative entropy between a quantum state and a set of quantum states plays a central role in entanglement theory \cite{brandao2011one}, where  relative entropies are  used to quantify the deviation   from the set of separable states, and in quantum thermodynamics \cite{horodecki2013fundamental,brandao2015second}, where  relative entropies are used to quantify the deviation from the set of Gibbs states.   In this paper we will extend the application of the max relative entropy  to dynamical scenarios, where  $\set S$ represents a set of quantum networks.   This extension is promising,~e.~g.~for applications to hypothesis testing. Indeed, it is natural to consider scenarios where one has a null hypothesis on the input-output behaviour of a quantum network and one wants to test the null hypothesis against an alternative hypothesis.   In the case of quantum states, the minimum probability of a type II error (failing to accept the alternative hypothesis) can be estimated in terms of the max relative entropy \cite{datta2013smooth}.   In the case of quantum networks, it is natural to expect that the max relative entropy defined here will yield  similar bounds---a result in this direction will be provided in Sections \ref{sec:maxcausal} and \ref{sec:maxnoncausal}.

\subsection{From semidefinite programs to the max relative entropy}  

In this section we provide a general bound on  the primal value of  an arbitrary semidefinite program. The bound can always be attained  and its value can be expressed in terms of a max relative entropy whenever the operator $A$ in the SDP $(\phi, A, B)$  is positive.    
To state the result, we need  some basic notation, provided in the following: 

    For a vector space $\spc V$, we denote by $\spc V^*$ the dual  space,~i.~e.~the space of linear functionals on $\spc V$.    Given  a subset $\set S\subseteq  \spc V$,  we  define the \emph{dual affine space}  $\overline {\set S}$
as  
\[    \overline {\set S}   :   =  \{  \Gamma \in  \spc V^*  ~|~    \<  \Gamma,  X\>   =  1  \, ,  \forall  X\in\set S     \}  \, .  \] 
Regarding $\spc V$ as a subspace of $\spc V^{**}$, one has the inclusion $\set S  \subseteq  \overline {\overline {\set S}}$.     When $\spc V$ is finite dimensional and $\set S$ is an affine set, one has the equality $\set S  =  \overline {\overline {\set S}}$.

Given a semidefinite program  $(\phi,A, B)$,  we define the \emph{primal affine space} as 
\begin{align}\label{primalaffine}    \set S
   :  =    \{  X  \in  \Herm (\spc X) ~|~    \phi  (X)   =  B  \} \, . 
\end{align}  
Simply,  $\set S$ is  the set of operators that satisfy the equality constraint of the primal problem. 
The \emph{dual affine space} is given by 
\begin{align}\label{dualaffine}  
 \overline  {  \set S}
  =     \{   \Gamma  \in  \Herm  (\spc X) ~|~   \<  \Gamma , X  \>   =  1\, , \forall X \in\set S
      \}  \, ,
\end{align}
having  used the identification of $\Herm (\spc X)$ with its dual space. With this notation, we have
\begin{theo}\label{lem:geomdual}
Let  $ (\phi, A, B)$ be a semidefinite program.   
The optimal solution of the primal problem is upper bounded as  
\begin{align}\label{geomdual} \omega_{\rm primal} \le   \inf_{\Gamma \in \overline  {\set S}}  \,  \min \{ \lambda \in  \R  ~|~  \lambda \Gamma  \ge A \}  \, ,  
\end{align}
where $\overline {\set S}$ is the dual affine space defined in Eq. (\ref{dualaffine}). 
If $\set S$ contains a positive operator and $\overline {\set S}$ contains a   strictly positive operator, then Eq.  (\ref{geomdual}) holds with the equality sign.   If, in addition, the operator $A$ is positive,  then one has the expression 
\begin{align}\label{maxentdual}
\omega_{\rm primal}    =   2^{  D_{\max}   ( A \,  \|\,   \overline S_+ \, )} \, ,
\end{align}
where $\overline {\set S}_+$ is the \emph{dual convex set}   $ \overline {\set S}_+   :=   \{\Gamma  \in \overline{\set S}   ~|~  \Gamma  \ge  0  \} $.  
\end{theo}
The proof can be found in appendix \ref{app:theo1}.

\medskip

We call the quantity $D  ( A  \,  \|  \,    \overline {\set S}_+)$ the \emph{max divergence  from normalization}. This quantity measures how much the operator $A$ deviates from the set of  positive functionals  that are normalized on every element of the primal set.      

The connection between semidefinite programming and the max relative entropy has previously appeared in the  special case where the task is to optimize quantum channels   \cite{konig-renner-2009-ieee,chiribella2013optimal}.     A related result was obtained by Jen\v cov\'a  in the framework of base norms \cite{jencova2014base}.   
In the next sections  we will elaborate on the physical meaning of Theorem \ref{lem:geomdual}, which will be  applied  to the optimization of quantum networks, both with definite and indefinite causal structure.   Before specializing ourselves to quantum networks, however, it is worth emphasizing a simple connection between the max relative entropy arising in generic SDPs   and the max relative entropy of quantum states.         

\begin{prop}\label{prop:maxoutput}
Let $C_0$ and $C_1$ be two elements of the convex set 
\[\set S_+   =   \{  X  \in \Herm  (\spc X)   \, | \, \phi  (X)   =  B \, ,   X\ge 0   \} \, .\] 
Then, one has the bound
\[       D_{\max}  \left (  \sqrt \Gamma   C_0   \sqrt \Gamma  \, \| \,     \sqrt \Gamma   C_1   \sqrt \Gamma  \right  )   \le   D_{\max}(  C_0 \,  \|     \,C_1)       \,  ,   \qquad \forall  \Gamma  \in \overline{\set S}_+    \, .  \]
The bound holds with the equality if the dual convex set $\overline{\set S}_+$ contains a full-rank operator. 
\end{prop} 
The proof can be found in  \ref{app:maxoutput}.  Note that, by construction  the operators $\sqrt \Gamma  C_i  \sqrt \Gamma\,, $  $i=0,1$ are density matrices: indeed, they are positive and  $\Tr[\sqrt \Gamma C_i \sqrt \Gamma]  =  \Tr [  \Gamma   C_i] =1$, since, by definition $\Gamma$ is a positive function normalized on the primal set $\set S$.   Proposition  will be used to show that the relative entropy of two  quantum networks is equal to the maximum relative entropy between  the output states generated by the networks.

\section{Optimizing quantum causal networks \label{convexopt}}
Here we consider the scenario where a network of quantum devices, arranged in a definite causal order,  is required to perform a desired task, such as implementing a distributed algorithm.   What is the maximum performance that the network can attain?   In this section we answer the question, measuring the performance through the score obtained in  a suitable test (depending on the task at hand) and providing a close form expression for the maximum score.

\subsection{The dual networks }

Following  Subsection \ref{subsec:test}, the mathematical description of the test is provided by a performance operator $\Omega$, acting on the Hilbert spaces of the input and output systems of the tested network.   The maximum performance achieved by an arbitrary causal network  is determined by  the following

 \begin{theo} \label{dual2}
Let $\Omega$ be an operator on $   \bigotimes_{j=1}^{N}   \left(  \,  \spc H^{\rm out}_j  \otimes \spc H_{j}^{\rm in} \right)$ and let $\omega_{\max}$  
  be the maximum of $\<\Omega,  C\> $ over all operators $C$ representing quantum  networks  of the form 
  \begin{align}\label{primalcircuit}
\Qcircuit @C=1em @R=0.2em @!R
{     & \qw \poloFantasmaCn{A_{1}^{\rm in}}  &  \multigate{1}{\map C_1}   &  \qw \poloFantasmaCn{A^{\rm out}_1}    &\qw &    & \qw\poloFantasmaCn{A_2^{\rm in}}  &  \multigate{1}{\map C_2}   & \qw \poloFantasmaCn{A_2^{\rm out}}   &  \qw  &\dots  &&\qw   \poloFantasmaCn{A_N^{\rm in}} &  \multigate{1}{\map C_N}  &  \qw  \poloFantasmaCn{A_N^{\rm out}}   &\qw \\
     & & \pureghost{\map C_1}  &   \qw &  \qw &\qw &\qw & \ghost{\map C_2}  &\qw  &  \qw &  \dots   &&  \qw & \ghost{\map C_N}    & &    }
 \quad .
\end{align}   
Then, $\omega_{\max}$ is given by 
  \begin{align}\label{deterministicprobability}
  \omega_{\max}    =    \min_{  \Gamma  \in  \Dual\Comb 
  }  \, \min  \{  \lambda  \in  \R ~|~   \lambda \Gamma  \ge   \Omega\}   \, ,
\end{align}
where    $\Dual\Comb$ denotes the set of \emph{dual combs}, that is,  positive operators $\Gamma$    representing networks of the form 
\begin{align}\label{dualcircuit} 
\Qcircuit @C=1em @R=0.2em @!R
{    & \pureghost{\sigma}  &   \qw  &\qw &\qw &\qw & \ghost{\map E_1}  &  \qw &\qw &  \dots   &&  \qw  & \ghost{\map E_{N-1}}  &   &   &    &   &  &\\
 & \pureghost{\sigma}  &     & & & & \pureghost{\map E_1}  &   & &    &&    & \pureghost{\map E_{N-1}}   & &    &   &     & & \\
      &  \multiprepareC{-2}{\sigma} &   \qw \poloFantasmaCn{A_{1}^{\rm in}}   & \qw     &     & \qw \poloFantasmaCn{A_{1}^{\rm out}}  &   \multigate{-2}{\map E_1}   &  \qw \poloFantasmaCn{A_{2}^{\rm in}}     &    \qw  &\dots  &  &  \qw \poloFantasmaCn{A_{N-1}^{\rm out}}   &  \multigate{-2}{\map E_{N-1}}  &    \qw \poloFantasmaCn{A_{N}^{\rm in}}     & \qw &        & \qw \poloFantasmaCn{A_{N}^{\rm  out}}    &     \measureD{\Tr   }    & \quad , } 
 \end{align} 
 where   $\sigma$ is a quantum state,    $(\map E_1,  \map E_2,\dots, \map E_{N-1})$ is  a sequence of quantum channels, and $\Tr_{A_N^{\rm out}}$ represents the trace over the last system.   Explicitly,  $\Dual\Comb$ is  the set of all positive operators $\Gamma$ satisfying the linear constraint 
\begin{align}
\Gamma  =& I_{A_N^{\rm out} }\otimes \Gamma^{(N)} \nonumber \\
\Tr_{A_n^{\rm in}}[\Gamma^{(N)}]=&I_{A^{\rm out}_{n-1}}\otimes \Gamma^{(N-1)} \ , \quad n=2,\ldots,N \nonumber \\
\label{trgamma1} \Tr_{A_1^{\rm in}}[\Gamma^{(1)}]=&1 \ , 
\end{align} 
for suitable  positive operators $\Gamma^{(n)}$ acting on  $ \spc H_n^{\rm in}  \otimes  \left[    \bigotimes_{j=1}^{n-1}   \left(  \,  \spc H^{\rm out}_j  \otimes \spc H_{j}^{\rm in} \right)\right]$.   
   When  $\Omega$ is positive, the  maximum performance can be expressed as 
\begin{align}\label{networkentropy}  \omega_{\max}    =  2^{  D_{\max}  \left  (\,   \Omega \, \|\,  \Dual\Comb 
\right)   }
\end{align} \, .
\end{theo}
The proof can be found in \ref{app:theo2}.

Theorem \ref{dual2} has an intuitive interpretation.    The dual networks (\ref{dualcircuit}) and the primal networks (\ref{primalcircuit})  ``deterministically complement each other": when two such networks  are connected, one obtains  the closed circuit 
\begin{align}\label{connected}
 \Qcircuit @C=1em @R=0.2em @!R
{ &    \multiprepareC{2}{\sigma}  & \qw  & \qw & \qw & \multigate{2}{\map E_1}  &\qw &  \qw &  \dots    &   & &  \qw & \multigate{2}{\map E_{N-1}}  &  &   & & &&& \\    
&    \pureghost{\sigma}  &   &  &  & \pureghost{\map E_1}  & &  &  \dots    &   & &   & \pureghost{\map E_{N-1}}  &  &     &  &    \\        
   &       \pureghost{\sigma}    &     \qw \poloFantasmaCn{A_{1}^{\rm in}}      &   \multigate{1}{\map C_1}   &   \qw \poloFantasmaCn{A_{1}^{\rm out}}    &    \ghost{\map E_1}    &    \qw \poloFantasmaCn{A_{2}^{\rm in}}     &\qw  & \dots  & &&     \qw \poloFantasmaCn{A_{N-1}^{\rm out}}      &  \ghost{\map E_{N-1}}  &    \qw \poloFantasmaCn{A_{N}^{\rm in}}       &     \multigate{1}{\map C_{N}}  &    \qw \poloFantasmaCn{A_{N}^{\rm out}}    &\measureD {\Tr}       & \quad   & \, ,  \\
   & &  &   \pureghost{\map C_1}  &   \qw &  \qw &   \qw &  \qw &\dots    && &\qw & \qw & \qw & \ghost{\map C_N}    &  &&& &}
\end{align}   
which   yields  no   information about the primal network and  makes any such information  inaccessible to further tests.  Hence, the dual networks represent the   non-informative tests. 
  The  max  relative entropy  quantifies how much the test with  performance operator $\Omega$ deviates from the set of non-informative tests.

\subsection{The case of binary testers}  

Consider a binary test, described by the tester $\{T_{\rm yes},  T_{\rm no}\}$ and assume that the test is passed if and only if the testing network yields the outcome 
``$\rm yes$".    Binary testers have applications in the theory of quantum interactive proof systems \cite{kitaev-watrous-2000}, where they can be used to compute the probability that the verifier accepts the token provided by the prover through a sequence of operations.   In this scenario,  the performance operator is given by  $\Omega  =   T_{\rm yes}$ and the probability that the prover  passes the test, optimized over all possible quantum strategies,  is  
\begin{align}\label{probmax}  p_{\max}    =  \left(  \,  \max_{\Gamma  \in  \Dual\Comb}   \max  \{   w  \,  T_{\rm yes}  \le  \Gamma  \}  \, \right)^{-1}
 \, ,   
 \end{align}
 having used Eq. (\ref{deterministicprobability}) with  $\lambda$ replaced  by its inverse $w  =  1/\lambda$. 
In  words, the problem is to find the maximum weight for which one can squeeze the tester operator $T_{\rm yes}$   under some dual comb   $\Gamma$. 

This maximization  has an intuitive interpretation: 
\begin{cor}
The maximum probability that a quantum causal network passes the test defined by the operator $T_1$ is equal to the inverse of the maximum weight $w$  for which   there exists a two-outcome tester  $\{  T_{\rm yes}',  T_{\rm no}'\}$  satisfying $T_{\rm yes}'   =  w\,  T_1$. 
\end{cor}  

\Proof    Suppose that the relation  $w  T_1 \le \Gamma$ holds for some weight  $w$ and some dual comb $\Gamma$.  Then, define $  T_{\rm yes}'  :=  w T_{\rm yes}$ and $T_{\rm no}'  : =      \Gamma  -  T_{\rm no}'$.    By construction, the operators $\{  T_{\rm yes}' ,  T_{\rm no}'\}$ form a tester:  they are positive and their sum satisfies Eq. (\ref{trtester}). \qed 

\medskip

  In other words,   the dual problem amounts to finding the binary tester $\{  T_{\rm yes}^*,T_{\rm no}^*\}$ that assigns the maximum possible probability to the  outcome  $1$, subject to the condition that $T_{\rm yes}^*$  is proportional to $T_{\rm yes}$.      The content of the duality is that the maximum  is attained when there exists a primal network that triggers deterministically the outcome $1$:  
      \begin{cor}
      Let $\{ T_{\rm yes}^*, T_{\rm no}^*\}$ be the optimal tester for the dual problem and let $ C^*$ be the  optimal quantum comb   for the primal problem.  Then, one has  
      \[     \<        T_{\rm yes}^*,  C^*\>   = 1  \, .  \]
      \end{cor}
\Proof    Let  $w^*$ be the optimal weight in the dual problem, Then, one  has  $T_{\rm yes}^*  =  w^*  \,  T_{\rm yes}$ and  $\<   T_{\rm yes}   ,     C^*\>   =  1/w^*$.  Combining these two equations, one gets       
    $   \<      T_{\rm yes}^* , C^*   \>    =      w^*  \,  \<   T_{\rm yes} ,  C^*  \>  =1$. \qed

  \medskip

\section{The conditional min-entropy of  quantum causal networks}\label{sec:conditional}

Theorem \ref{dual2} allows us to extend the notion of conditional min-entropy      \cite{datta-renner-2009} from quantum states to quantum causal networks.   Let us first review the basic properties of the conditional min-entropy of  quantum states: For a quantum state $\rho  \in  \St (AB)$,  the conditional min-entropy of system $A$, conditional on system $B$, is defined as     \cite{datta-renner-2009} 
\begin{align}\label{minentropy}  H_{\rm min}    (A|B)_\rho    :=    -  \log  \left[    \min_{\gamma  \in   \St (B)}  \min  \{  \lambda    \in \R ~|~   \lambda  ( I_A  \otimes  \gamma_B   )   \ge \rho_{AB} \}      \right]  \, .\end{align} 
K\"onig, Renner, and Schaffner \cite{konig-renner-2009-ieee}  clarified the operational meaning of $  H_{\rm min}    (A|B)_\rho$  in terms of the following task:  given the state $\rho_{AB}$, find the quantum channel $\map C$ that produces the best approximation of the maximally entangled state $|\Phi\>_{A \, A'}  : =   \sum_{n=1}^{d_A}  \,  |n\>|n\>/\sqrt d_A$, by acting locally on  system $B$.    Here the quality of the approximation is measured by the fidelity, namely the probability that the output state passes a binary test with POVM $\{  P_{\rm yes},  P_{\rm no}\}$,  defined by $P_{\rm yes}  :=  |\Phi\>\<\Phi|$.    
  Overall, we can jointly regard the preparation of the state $\rho$ and the measurement of the binary POVM $\{  P_{\rm yes},  P_{\rm no}\}$ as a test performed on the channel  $\map C$.  Diagrammatically, the successful instance of the test  is represented by the network  
\begin{align}\label{rhotest}
 \Qcircuit @C=1em @R=0.2em @!R
{   &\multiprepareC{1}{\rho}  &  \qw \poloFantasmaCn{A}    & \qw &   \qw  &    \qw  &   \multimeasureD{1}{   P_{\rm yes}  } \\
&\pureghost{\rho}   &    \qw \poloFantasmaCn{B}   &\qw &  &   \qw \poloFantasmaCn{A'}   &   \ghost{    P_{\rm yes}  } }
 \quad ,
 \end{align}
 whose Choi operator is given by 
 \begin{align*}  T_{\rm yes}  &:=   \rho  *  P^T_{\rm yes}  \\
   &  =   \rho/d_A \, ,
   \end{align*} 
   (with a slight abuse of notation, in the second equality we regard $\rho$ as an operator on $A' B$, instead of $AB$).
 Hence, the probability that the channel passes the test is 
 \begin{align*}
 p  &   =      T_{\rm yes} *    C    \\
 &   =  \frac{ \Tr\left[  \rho \, C^T\right]}{d_A} \, ,     
 \end{align*}
 where $C$ is the Choi operator of $\map C$. 
 K\"onig, Renner, and Schaffner showed that the maximum probability over all possible channels is 
 \begin{align}p_{\max}   =    \frac{2^{-  H_{\min}  (A|B)_\rho} }{d_A}  \, .\end{align} \label{entropystates}
We now extend the  notion of conditional min-entropy   from states to  networks with a definite causal structure.   This can be done in two slightly different ways, illustrated in the following subsections.

\subsection{The conditional min-entropy of a quantum causal network}

 The first way to generalize the conditional min-entropy from states is to regard $H_{\min} (A|B)_\rho$ as a measure of the correlations that can be extracted from the state $\rho_{AB}$ by acting on system $B$ alone.     A natural generalization to the network scenario arises if we  consider a quantum network of the form
 \begin{align}\label{minentropynetwork}
\Qcircuit @C=1em @R=0.2em @!R
{   
     & & \pureghost{\map D_1}  &   \qw &  \qw &\qw &\qw & \ghost{\map D_2}  &\qw  &  \qw &  \dots   &&  \qw & \ghost{\map D_N}    & &  \\
      & & \pureghost{\map D_1}  &    &  & & & \pureghost{\map D_2}  & &  &    &&  &    & &  \\
       & \qw \poloFantasmaCn{B_{1}^{\rm in}}  &  \multigate{-2}{\map D_1}   &  \qw \poloFantasmaCn{B^{\rm out}_1}    &\qw &    & \qw\poloFantasmaCn{B_2^{\rm in}}  &  \multigate{-2}{\map D_2}   & \qw \poloFantasmaCn{B_2^{\rm out}}   &  \qw  &\dots  &&\qw   \poloFantasmaCn{B_N^{\rm in}} &  \multigate{-2}{\map D_N}  &  \qw  \poloFantasmaCn{B_N^{\rm out}}   &\qw    }
 \quad , 
\end{align}
and ask how much correlation can be generated by interacting with the network in  the first $N-1$ time steps.  To generate the correlations, we can connect the network (\ref{minentropynetwork}) with a second network that processes  all input/output systems  before $  B_N^{\rm out}$. Graphically, the second network can be described as
\begin{align}\label{scorer}
\Qcircuit @C=1em @R=0.2em @!R
{     &  \multiprepareC{2}{\sigma}   &  \qw \poloFantasmaCn{B^{\rm in}_1}    &\qw &    & \qw\poloFantasmaCn{B_1^{\rm out}}  &  \multigate{2}{\map E_1}   & \qw \poloFantasmaCn{B_2^{\rm in}}   &  \qw  &\dots  &&\qw   \poloFantasmaCn{B_{N-1}^{\rm in}} &  \multigate{2}{\map E_{N-1}}  &  \qw  \poloFantasmaCn{B_N^{\rm in}}   &\qw \\
   &  \pureghost{\sigma}  &    &   & & & \pureghost{\map E_1}  &  &   &    &&   & \pureghost{\map C_{E-1}}    &    & \\ 
     &  \pureghost{\sigma}  &   \qw &  \qw &\qw &\qw & \ghost{\map E_1}  &\qw  &  \qw &  \dots   &&  \qw & \ghost{\map C_{E-1}}    &  \qw  \poloFantasmaCn{B_N^{\rm out'}}   &\qw    } \, ,
\end{align} 
where $B_N^{\rm out'}$ is a quantum system of the same dimension as $B_N^{\rm out}$. 
When the two networks are connected, they generate  the bipartite state
\begin{align}\label{state}
 \Qcircuit @C=1em @R=0.2em @!R
{   
     & & \pureghost{\map D_1}   &  \qw &\qw &\qw & \ghost{\map D_2}  &\qw  &  \qw &  \dots  & & \qw &\qw&  \qw  & \ghost{\map D_N}    & &  \\
      & & \pureghost{\map D_1}      &  & & & \pureghost{\map D_2}  & &  &   & &&  &  &  & &  \\
        \multiprepareC{2}{\sigma}  & \qw \poloFantasmaCn{B_{1}^{\rm in}}  &  \multigate{-2}{\map D_1}   &  \qw \poloFantasmaCn{B^{\rm out}_1}     & \multigate{2}{\map E_1}   & \qw\poloFantasmaCn{B_2^{\rm in}}  &  \multigate{-2}{\map D_2}   & \qw \poloFantasmaCn{B_2^{\rm out}}   &  \qw  &\dots   & &  \qw  \poloFantasmaCn{B_{N-1}^{\rm out}}    &  \multigate{2}{\map E_{N-1}}&\qw   \poloFantasmaCn{B_N^{\rm in}} &  \multigate{-2}{\map D_N}  &  \qw  \poloFantasmaCn{B_N^{\rm out}}   &\qw    \\
          \pureghost{\sigma}  & &  &     & \pureghost{\map E_1}   &     &   &    &  &   & &    &\pureghost{\map E_{N-1}}&  &   &   &   \\
               \pureghost{\sigma}  & \qw   & \qw    &  \qw     & \ghost{\map E_1}   & \qw   & \qw  & \qw    &  \qw  &\dots   & &  \qw     &  \ghost{\map E_{N-1}}&\qw     \poloFantasmaCn{B_N^{\rm out'}} & \qw  &  \qw    &\qw    \\ }
 \quad .
\end{align}
A measure of the correlations generated by the interaction of the two networks  is then provided by the fidelity between the  state  (\ref{state}) and the maximally entangled state. 
Explicitly, the fidelity is given by
\begin{align} 
F  :  &  =    \<   \Phi|  \,   \left(   \sigma     *  D_1   *  E_1* \cdots   *  D_{N-1}  * E_{N-1}   *  D_N  \right)  \,  |\Phi\>  \nonumber \\
  &  =   \frac{ \Tr  \left[  D \,  E^T  \right] }{d_{B^{\rm out}_N}}\, , \label{Fidelity}
\end{align}
with 
\[D:  =  D_1*  \cdots * D_N \qquad {\rm  and } \qquad E :  =   \sigma *  E_1 *\cdots *  E_{N-1}    \]  
 (with a little abuse of notation, in the second equality we regard $E$ as an operator on $\spc H_{B_1^{\rm in}}  \otimes \spc H_{B_1^{\rm out}}  \otimes \cdots  \otimes \spc H_{B_N^{\rm in}} \otimes \spc H_{B_N^{\rm out }} $ instead of $\spc H_{B_1^{\rm in}}  \otimes \spc H_{B_1^{\rm out}}  \otimes \cdots  \otimes \spc H_{B_N^{\rm in}} \otimes \spc H_{B_N^{\rm out \, \prime }} $).

The maximum of the fidelity  over all networks of the form (\ref{scorer})  can be computed via Theorem \ref{dual2}, which yields the expression 
  \begin{align} \label{Fmax} F_{\max}   =   \frac{ \min_{\Gamma_{t_1\dots  t_{N-1}}    }  \min    \left\{  \lambda    \in \R ~\left|~   \lambda  \left( I_{t_N}  \otimes  \Gamma_{t_1\dots t_{N-1}}   \right)   \ge R  \right\}  \right.   }{ d_{B^{\rm out}_{N}}}  \, ,
   \end{align}
where    $\Gamma_{t_1  \dots  t_{N-1}}$ is a generic element of $ \Comb\left(  B_1^{\rm in} \to  B_1^{\rm out} , \, \dots \, ,  \,    B_{N-1}^{\rm in} \to  B_{N-1}^{\rm out} \right)$  and  $I_{t_N}   :  =  I_{B_N^{\rm out}}\otimes I_{B_N^{\rm in}}$.   

Eq. (\ref{Fmax}) motivates the following 
\begin{defi}\label{def:entromino}
Let $D  \in  \Comb (  B_1^{\rm in}  \to  B_1^{\rm out} \, , \dots \, , B_N^{\rm in}  \to B_N^{\rm out}  )$ be a quantum comb and let $t_j   :=     B_j^{\rm in}  \to  B_j^{\rm out} $ be the type  corresponding to the $j$-th time step.  
The \emph{network min-entropy}  of the  $N$-th time step, conditionally on the first $N-1$ time steps  is the quantity
\begin{align}
\nonumber   & H_{\rm min}  (  t_N  \, | \,    t_1\cdots  t_{N-1})_D \\
 \label{networkminentropy}& \quad :    =  -  \log  \left[    \min_{\Gamma_{t_1\dots  t_{N-1}}    }  \min    \left\{  \lambda    \in \R ~\left|~   \lambda  \left( I_{t_N}  \otimes  \Gamma_{t_1\dots t_{N-1}}   \right)   \ge D  \right\}  \right.      \right]  \, ,  \qquad \qquad \qquad 
\end{align}
where the first minimum is over the elements of $ \Comb\left(  B_1^{\rm in} \to  B_1^{\rm out} , \, \dots \, ,  \,    B_{N-1}^{\rm in} \to  B_{N-1}^{\rm out} \right)$  and  $I_{t_N}   :  =  I_{B_N^{\rm out}}\otimes I_{B_N^{\rm in}}$. 
\end{defi}



The above definition is a compelling generalization of the conditional min-entropy for states.   First of all, it comes with a natural operational interpretation, as the maximum  amount of correlations between the last output of the network and all the system involved in the previous history.    Moreover, the conditional min-entropy of quantum networks is consistent with the conditional min-entropy of quantum states: Concretely, one can interpret the conditional min-entropy   (\ref{networkminentropy}) as the maximum conditional min-entropy of the output  state of the network, conditionally on an external reference system generated through the intermediate time steps.    This interpretation is based on the following
\begin{prop}\label{entropymotivation}
For a causal network with Choi operator $D  \in  \Comb (  B_1^{\rm in}  \to  B_1^{\rm out} \, , \dots \, , B_N^{\rm in}  \to B_N^{\rm out}  )$,  the min-entropy   $H_{\rm min}  (  t_N  \, | \,    t_1\cdots  t_{N-1})_D$  is equal to the min-entropy of the output state $\rho \in   \St \left(  \mathcal{H}_{B_N^{\rm out}}  \otimes  \mathcal{H}_{B_N^{out  \, \prime}} \right)$ in  Eq. (\ref{state}) maximized over the input state $\sigma$ and over  the sequence of intermediate operations $\map E_1,\dots,  \map E_{N-1}$.
\end{prop}
The proof is given in \ref{app:entropyproof}.  We expect that the network min-entropy defined in Eq. (\ref{networkminentropy})    will  play a  role in the study non-Markovian quantum evolutions,  along the lines of the  entropic  characterization of Markovianity  provided in Refs. \cite{buscemi2016equivalence,bae2016operational}. 
  Intuitively, the idea is that one can evaluate how the correlations build up from one step to the next and use this information to infer properties of the internal memory used by the network.

   \medskip 
 
\subsection{The conditional min-entropy of a test}  
An alternative  way to extend the notion of conditional min-entropy is to regard $H_{\min}  (A|B)_\rho$ as a quantity associated to a test---specifically, the test depicted    in Eq. (\ref{rhotest}).   From this point of view, it is natural to extend the definition to tests consisting of multiple time steps, as follows
  \begin{defi}\label{def:minentrop}
Let $T_{\rm yes}$ be a positive operator associated to a test of the form
  \begin{align}\label{testerfinal}
   \Qcircuit @C=1em @R=0.2em @!R
{    & \pureghost{\rho}  &   \qw  &\qw &\qw &\qw & \ghost{\map D_1}  &  \qw &\qw &  \dots   &&  \qw  & \ghost{\map D_{N-1}}  &\qw   & \qw   &  \qw  &    \qw & \ghost{P_{\rm yes}}\\
 & \pureghost{\rho}  &     & & & & \pureghost{\map D_1}  &   & &    &&    & \pureghost{\map D_{N-1}}   & &    &   &     & \pureghost{  P_{\rm yes}}\\
      &  \multiprepareC{-2}{\rho} &   \qw \poloFantasmaCn{A_{1}^{\rm in}}   & \qw     &     & \qw \poloFantasmaCn{A_{1}^{\rm out}}  &   \multigate{-2}{\map D_1}   &  \qw \poloFantasmaCn{A_{2}^{\rm in}}     &    \qw  &\dots  &  &  \qw \poloFantasmaCn{A_{N-1}^{\rm out}}   &  \multigate{-2}{\map D_{N-1}}  &    \qw \poloFantasmaCn{A_{N}^{\rm in}}     & \qw &        & \qw \poloFantasmaCn{A_{N}^{\rm  out}}    &     \multimeasureD{-2}{  P_{\rm yes  } } } 
\quad . 
 \end{align}  
 The conditional min-entropy of the output system $A^{\rm out}_N$, conditionally on all the previous systems is 
\begin{align}\label{abcd}
H_{\rm min}  (  A^{\rm out}_N  \, | \,    A^{\rm in}_1  A^{\rm out}_1 A^{\rm in}_2  \dots    A^{\rm in}_N  )_{T_{\rm yes}}     :    =  -  \log  \left[    \min_{  \Gamma^{(N)}      }  \min    \left\{  \lambda    \in \R ~\left|~   \lambda  \left( I_{A^{\rm out}_N}  \otimes  \Gamma^{(N)}\right)   \ge T_{\rm yes}  \right\}  \right.      \right]  \, ,
\end{align}
where $  \Gamma^{(N)}$ is a generic element of $ \Comb  \left(  I \to  A_1^{\rm in} , \,    A^{\rm out}_1  \to A^{\rm in}_2   \dots \, ,  \,    A_{N-1}^{\rm out} \to  A_{N}^{\rm in} \right)$, corresponding to a network of the form 
\begin{align} 
\Qcircuit @C=1em @R=0.2em @!R
{    & \pureghost{\sigma}  &   \qw  &\qw &\qw &\qw & \ghost{\map E_1}  &  \qw &\qw &  \dots   &&  \qw  & \ghost{\map E_{N-1}}  &  & \\
 & \pureghost{\sigma}  &     & & & & \pureghost{\map E_1}  &   & &    &&    & \pureghost{\map E_{N-1}}   &  &\\
      &  \multiprepareC{-2}{\sigma} &   \qw \poloFantasmaCn{A_{1}^{\rm in}}   & \qw     &     & \qw \poloFantasmaCn{A_{1}^{\rm out}}  &   \multigate{-2}{\map E_1}   &  \qw \poloFantasmaCn{A_{2}^{\rm in}}     &    \qw  &\dots  &  &  \qw \poloFantasmaCn{A_{N-1}^{\rm out}}   &  \multigate{-2}{\map E_{N-1}}  &    \qw \poloFantasmaCn{A_{N}^{\rm in}}     & \qw  } 
 \quad . 
 \end{align}  
\end{defi}
The conditional min-entropy for ``states"  [or, more precisely, for tests of the form (\ref{rhotest})] can be retrieved as a special case of this definition, by setting $N=1$,   $A^{\rm in}_1  =  B$, $A^{\rm out}_1  =  A  $, and $T_{\rm yes}  =  \rho/d_A$. The appeal of the above  definition is that  it extends the definition of min-entropy to a class of  probabilistic operations.

The conditional min-entropy can be interpreted operationally as the (negative logarithm of the) maximum probability that a quantum causal network passes  the test $T_{\rm yes}$. This interpretation follows from  Theorem \ref{dual2}, which yields  the following 
\begin{cor} The maximum probability that a quantum causal network  of the form 
  \begin{align*}
  \Qcircuit @C=1em @R=0.2em @!R
{     & \qw \poloFantasmaCn{A_{1}^{\rm in}}  &  \multigate{1}{\map C_1}   &  \qw \poloFantasmaCn{A^{\rm out}_1}    &\qw &    & \qw\poloFantasmaCn{A_2^{\rm in}}  &  \multigate{1}{\map C_2}   & \qw \poloFantasmaCn{A_2^{\rm out}}   &  \qw  &\dots  &&\qw   \poloFantasmaCn{A_N^{\rm in}} &  \multigate{1}{\map C_N}  &  \qw  \poloFantasmaCn{A_N^{\rm out}}   &\qw \\
     & & \pureghost{\map C_1}  &   \qw &  \qw &\qw &\qw & \ghost{\map C_2}  &\qw  &  \qw &  \dots   &&  \qw & \ghost{\map C_N}    & &    }
 \quad .
\end{align*}     passes the test with operator $T_{\rm yes}$ is 
 \begin{align*}
 p_{\max}   
&  =    2^{- H_{\rm min}  (  A^{\rm out}_N  \, | \,    A^{\rm in}_1  A^{\rm out}_1 A^{\rm in}_2  \dots    A^{\rm in}_N  )_{T_{\rm yes}}     } \, .
\end{align*} 
  \end{cor}
This result secures an operational interpretation for the conditional min-entropy defined in Eq. (\ref{abcd}).  Quite intuitively, the conditional min-entropy of the test  is a measure of how of the first $N-1$ time steps can be used to predict the outcome of the measurement performed in the last step.

\section{The max relative entropy of  causal networks}\label{sec:maxcausal}  
We conclude our study of causal networks with a result relating  the max relative entropy of quantum networks to the max relative entropy of quantum states: 
\begin{prop}\label{prop:maxrelativecausal}
Let $C^{(0)}$ and $C^{(1)}$ be the Choi operators of two  networks  
 \begin{align}
 \Qcircuit @C=1em @R=0.2em @!R
{     & \qw \poloFantasmaCn{A_{1}^{\rm in}}  &  \multigate{1}{\map C^{(x)}_1}   &  \qw \poloFantasmaCn{A^{\rm out}_1}    &\qw &    & \qw\poloFantasmaCn{A_2^{\rm in}}  &  \multigate{1}{\map C^{(x)}_2}   & \qw \poloFantasmaCn{A_2^{\rm out}}   &  \qw  &\dots  &&\qw   \poloFantasmaCn{A_N^{\rm in}} &  \multigate{1}{\map C^{(x)}_N}  &  \qw  \poloFantasmaCn{A_N^{\rm out}}   &\qw \\
     & & \pureghost{\map C^{(x)}_1}  &   \qw &  \qw &\qw &\qw & \ghost{\map C^{(x)}_2}  &\qw  &  \qw &  \dots   &&  \qw & \ghost{\map C^{(x)}_N}    & &    }\, ,\qquad x  =  0,1 \, ,
\end{align} 
and let $E$ be the Choi operator of a network of the form
\begin{align}\label{scorer1}
\Qcircuit @C=1em @R=0.2em @!R
{    &  \pureghost{\sigma}  &   \qw &  \qw &\qw &\qw & \ghost{\map E_1}  &\qw  &  \qw &  \dots   &&  \qw & \ghost{\map C_{E-1}}    &  \qw  \poloFantasmaCn{  S}   &\qw \\
&  \pureghost{\sigma}  &    &   & & & \pureghost{\map E_1}  &  &   &    &&   & \pureghost{\map C_{E-1}}    &    & \\ 
       &  \multiprepareC{-2}{\sigma}   &  \qw \poloFantasmaCn{A^{\rm in}_1}    &\qw &    & \qw\poloFantasmaCn{A_1^{\rm out}}  &  \multigate{-2}{\map E_1}   & \qw \poloFantasmaCn{A_2^{\rm in}}   &  \qw  &\dots  &&\qw   \poloFantasmaCn{A_{N-1}^{\rm in}} &  \multigate{-2}{\map E_{N-1}}  &  \qw  \poloFantasmaCn{A_N^{\rm in}}   &\qw   } \, ,
\end{align} 
where $S$ is a generic quantum system.    Then, one has  
\begin{align}\label{dmaxstates}
D_{\max}   (  C^{(0)} \, \| \,  C^{(1)})   =  \max_E   \, D_{\max}   (     C^{(0)}  *E \, \| \,  C^{(1)}  * E) \, . 
\end{align}
where the maximum runs over all networks of the form (\ref{scorer1}), with arbitrary system $S$. 
\end{prop}  
The proof is provided in \ref{app:maxrelativecausal}. 
In words, Eq. (\ref{dmaxstates}) states that the max relative entropy between two quantum networks is equal to the max relative entropy between the output states one can generate from them.  Diagrammatically, the output states are 
\begin{align}\label{connected}
 \Qcircuit @C=1em @R=0.2em @!R
{ &    \multiprepareC{2}{\sigma}  & \qw  & \qw & \qw & \multigate{2}{\map E_1}  &\qw &  \qw &  \dots    &   & &  \qw & \multigate{2}{\map E_{N-1}}  &   \qw \poloFantasmaCn{S} &\qw   &\qw & \qw&  & & \\    
&    \pureghost{\sigma}  &   &  &  & \pureghost{\map E_1}  & &  &  \dots    &   & &   & \pureghost{\map E_{N-1}}  &  &     &  &    \\        
   &       \pureghost{\sigma}    &     \qw \poloFantasmaCn{A_{1}^{\rm in}}      &   \multigate{1}{\map C^{(x)}_1}   &   \qw \poloFantasmaCn{A_{1}^{\rm out}}    &    \ghost{\map E_1}    &    \qw \poloFantasmaCn{A_{2}^{\rm in}}     &\qw  & \dots  & &&     \qw \poloFantasmaCn{A_{N-1}^{\rm out}}      &  \ghost{\map E_{N-1}}  &    \qw \poloFantasmaCn{A_{N}^{\rm in}}       &     \multigate{1}{\map C^{(x)}_{N}}  &    \qw \poloFantasmaCn{A_{N}^{\rm out}}    &\qw     & \qquad   &  \qquad   \,   x =  0,1 .  \\
   & &  &   \pureghost{\map C^{(x)}_1}  &   \qw &  \qw &   \qw &  \qw &\dots    && &\qw & \qw & \qw & \ghost{\map C^{(x)}_N}    &  &&& &}  
\end{align} 
Proposition \ref{prop:maxrelativecausal} has an application to problems of hypothesis testing where the task is to distinguish between two quantum networks. Here one has access to an quantum network, that is promised to have quantum comb $C^{(0)}$ or $C^{(1)}$.     In order to determine which of these two hypotheses is correct, one has to interact with the network, by sending inputs to it and processing its outputs.  In the end, these operations will result in the preparation of a quantum state, as in Eq. (\ref{connected}).  At this point, the problem is to distinguish between two states $\rho^{(0)}$ and $\rho^{(1)}$ corresponding to the two hyoptheses.       One-shot    hypothesis testing of quantum states has been studied by Datta \emph{et al} in Ref. \cite{datta2013smooth}, where they provided bounds on the type II error probability in terms of the max relative entropy.  Proposition  \ref{prop:maxrelativecausal} then allows to relate the max relative entropy of the output states to the max relative entropy of the networks, opening a route to adapting the results of  \cite{datta2013smooth} to the study of hypothesis testing to the more general  scenario.

\section{Non-causal networks \label{non-causal}}
In the previous sections we restricted our focus to causal networks. We will  address  the  general scenario, concerning networks that are not compatible with any pre-defined causal order   \cite{hardy-2009-book,chiribella-2009-arxiv,oreshkov-2012-nature,chiribella-2012-pra,chiribella-dariano-2013-pra,colnagh-2012-pla,baumeler-2013-arxiv,araujo-2014-prl,baumeler-wolf-2016}.    Some of these networks arise when multiple quantum devices are connected in a way that is controlled by the state of a quantum system \cite{chiribella-2009-arxiv,colnagh-2012-pla}. Some other networks are not built from individual devices  \cite{oreshkov-2012-nature,chiribella-dariano-2013-pra} but may  possibly arise in exotic quantum gravity scenarios.  These generalized quantum networks are characterized by the way in which they interact with external quantum devices.  

\subsection{A bipartite example}
The characterization of the non-causal networks is not as simple as in the case of causally ordered networks.  We first  illustrate the idea   in a simple example, inspired by the  work  of Oreshkov, Costa, and Brukner \cite{oreshkov-2012-nature}.    
  Imagine two laboratories, $A$ and $B$, where two parties, Alice and Bob perform local experiments.    In each laboratory, ordinary quantum theory holds and, in particular, one can describe the time evolution by quantum channels. Specifically, let $\map A$ and $\map B$ be the quantum channels describing the evolution of the systems in laboratories $A$ and $B$, respectively.    Now, one can model the interactions between one laboratory and the other  by a generalized quantum network, which describes the background structure of spacetime.  
  
  Concretely, suppose that, at some earlier time, system $A_1^{\rm in}$ in the first laboratory has been prepared jointly with system $B_1^{\rm in}$ in the second laboratory, and that, at a later time, system $A_1^{\rm out}$ and system $B_1^{\rm out}$ are discarded.  Indulging into a  bit of science fiction, one could  imagine a scenario where systems $A_1^{\rm in}$ and $B_1^{\rm in}$ emerge from a wormhole at time $t_0$ and system $A_1^{\rm out}$ and $B_1^{\rm out}$  enter the same wormhole at time $t_1$.  Between times $t_0$ and $t_1$ the systems $A_1^{\rm in}$ and $B_1^{\rm in}$ can interact with the other systems in Alice's and Bob's laboratories, here denoted as $A^{\rm in}_2, A_2^{\rm out}$ and $B^{\rm in}_2, B^{\rm out}_2$, respectively. The interaction is controlled locally by Alice and Bob, who implement the channels $\map A$ and $\map B$, as illustrated in Figure \ref{brukner+}.  The connection of Alice's and Bob's laboratories through the background spacetime structure  can be described as a map  
 \begin{align}
 \map S  :     \map A \otimes \map B    \mapsto   \map S (\map A\otimes \map B) \, , 
 \end{align} 
 which transforms the quantum channels $\map A$ and $\map B$ into a new quantum channel $\map S (\map A\otimes  \map B)$.    
   Maps that transform channels into channels are known as  \emph{quantum supermaps}  \cite{chiribella-dariano-2008-epl,chiribella-dariano-2009-pra}.  The basic requirements for quantum supermaps  are linearity, complete positivity,  and normalization.   In this setting, linearity means that one has 
    \begin{align}
   \map S  \left(   \sum_{i}  \,   p_i  \,  \map A_i \otimes   \map B_i\right)   =  \sum_{i}  \,  p_i   \,  \map S  ( \map A_i\otimes \map B_i) \, ,
   \end{align}  
   for every choice of coefficients $\{ p_i\}$.    The standard motivation for linearity comes from the requirement that convex combinations of input channels   (generated by Alice and Bob by sharing random bits) be mapped into convex combinations of the corresponding outputs.  
   
\begin{figure}
\begin{center}
\hspace*{0cm}\includegraphics[scale=.5]{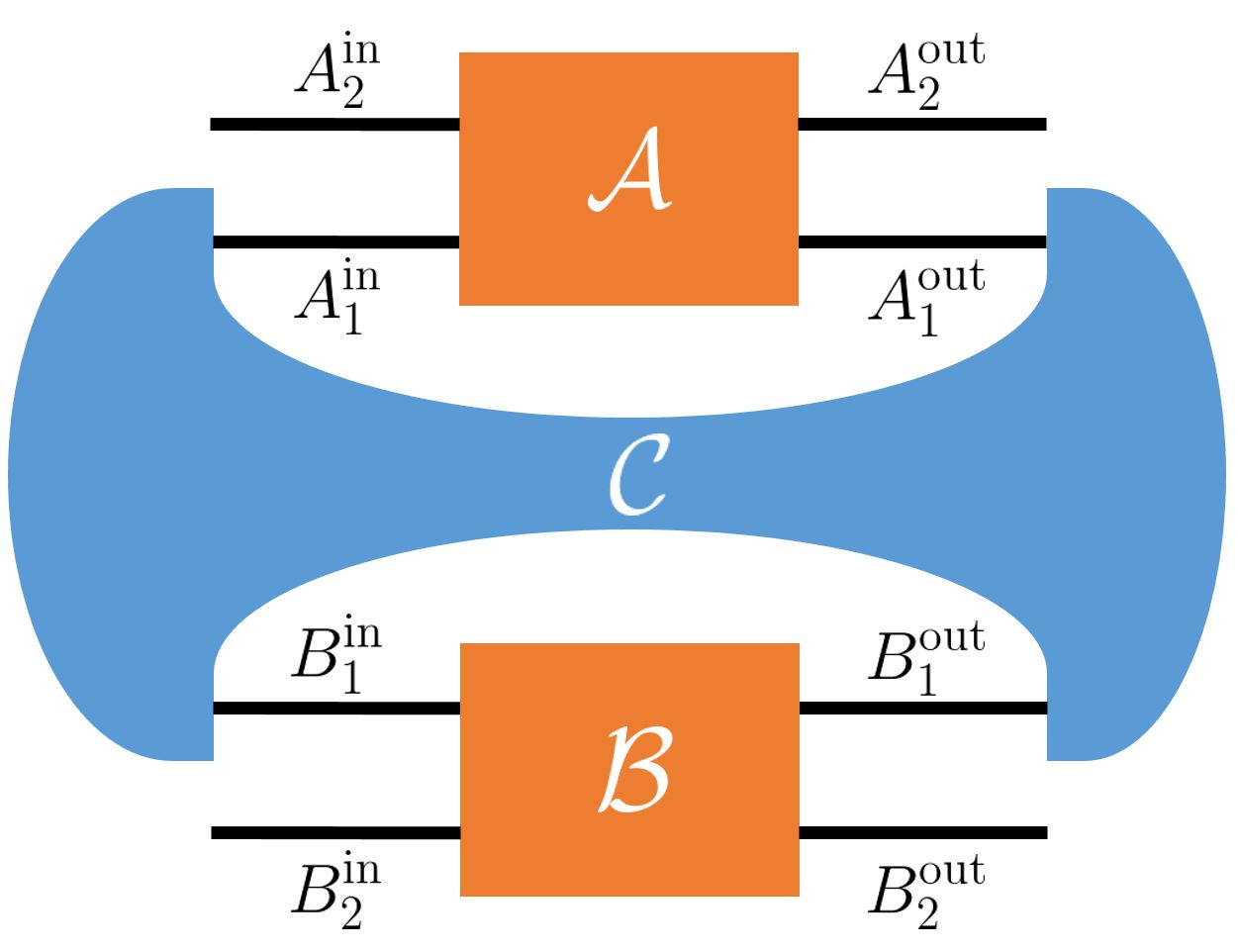}
\end{center}
\caption{The quantum channels $\map A$ and $\map B$ in Alice's and Bob's laboratory interact through a  quantum network $\map C$, describing the interactions mediated by the background spacetime. Here, $\map A$ ($\map B$) is a  bipartite channel transforming the input systems $A_1^{\rm in} A_2^{\rm in}$   ($B_1^{\rm in} B_2^{\rm in}$) into the output systems $A_1^{\rm out} A_2^{\rm out}$  ($B_1^{\rm out} B_2^{\rm out}$).  The connection between the channels take places only through the systems $A^{\rm in}_1,  A^{\rm out}_1,  B_1^{\rm in}$ and $B_1^{\rm out}$, while systems  $A^{\rm in}_2,  A^{\rm out}_2,  B_2^{\rm in}$ and $B_2^{\rm out}$ do not interact directly. }
\label{brukner+}
\end{figure}

Regarding complete positivity,  it can be motivated by the local form of the interactions.     Since the interaction between Alice's and Bob's laboratory takes place only through systems $A_1$ and $B_1$, it is natural to assume that the supermap $\map S$  acts non-trivially only on these systems, as  
 \begin{align}\label{localsupermap} 
 \map S  (  \map A\otimes \map B)    =     (\map I_{A_2\to  A_2}  \otimes \map C  \otimes \map I_{B_2\to B_2})  (  \map A  \otimes \map B) \, ,
 \end{align}
   where $\map I_{A_2\to  A_2}$   ($\map I_{B_2\to  B_2}$) is the \emph{identity supermap},  acting trivially on the channels with input $A_2$ ($B_2$) and output $A_2$  ($B_2$), and $\map C$ is a supermap that  annihilates   channels with input $A_1B_1$ and output $A_1B_1$.  Physically, the  map $\map C$    represents the   piece of spacetime connecting  $A^{\rm in}_1,  A^{\rm out}_1,  B_1^{\rm in}$ and $B_1^{\rm out}$,$A^{\rm in}_2,  A^{\rm out}_2,  B_2^{\rm in}$ and $B_2^{\rm out}$ with $A_1'$ and $B_1'$.   

\subsection{Choi operator formulation}    
Since all the maps  $\map A, \map B$, and $\map C$ are completely positive, one can represent them with Choi operators $A$,  $B$, and $C$,  respectively.  In terms of Choi operators, Eq. (\ref{localsupermap}) can be expressed as 
\begin{align}\label{localsupermapchoi}
C'   =\Tr_{  A_1^{\rm in},  A_1^{\rm out}, B_1^{\rm in}, B_1^{\rm out} }    \left [     (   I_{A_2^{\rm out}}  \otimes I_{A_2^{\rm in}}  \otimes  C  \otimes  I_{B_2^{\rm out}}  \otimes   I_{ B_2^{\rm in}})   (  A^T  \otimes B^T)  \right] \, ,
\end{align}
where  $C'$  the Choi operator of the channel $\map S (\map A\otimes \map B)$. 
 Oreshkov, Costa, and Brukner  refer to the Choi operator $C$ as a \emph{process matrix} \cite{oreshkov-2012-nature}.  For the supermaps that can be implemented  by connecting quantum devices in a fixed causal structure,   Choi operators $C$ are the same as the quantum combs considered in the previous sections.

Here the operator $C$ acts on the tensor product Hilbert space   $ \spc H_{A_1}^{\rm in} \otimes \spc H_{A_1}^{\rm out}  \otimes \spc H_{B_1}^{\rm in} \otimes \spc H_{B_1}^{\rm out}$.   In order to be the Choi operator of a valid quantum network, the operator $C$ must be positive semidefinite and satisfy a suitable normalization condition---specifically, $C$ should satisfy the condition 
\begin{align}\label{normalizationchoi}
\Tr_{   A_1^{\rm in},  A_1^{\rm out}, B_1^{\rm in}, B_1^{\rm out} }    \left [        C    ( \widetilde A  \otimes  \widetilde B)  \right]    =    1 
\end{align}
for every operators  $\widetilde A$ and $\widetilde B$ satisfying the conditions  
\begin{align}
\Tr_{A_1^{\rm out}}[\widetilde A]=   I_{A_1^{\rm in}}   \qquad {\rm and}    \qquad \Tr_{B_1^{\rm out}}[\widetilde B]=   I_{B_1^{\rm in}}  \, .
\end{align}
(See appendix  \ref{app:condition} for the derivation).  Physically, this means that  the non-causal network $\map C$ deterministically annihilates every pair of local channels $\widetilde A$ and $\widetilde B$, acting on systems $A_1^{\rm in}$, $A_1^{\rm out}$ and $B_1^{\rm in}$, $B_1^{\rm out}$, respectively. 

Equivalently, the valid networks  can be characterized as in the following: 
\begin{prop}\label{prop:N}
An operator $C$ is the Choi operator of a non-causal network as  in  Fig.  \ref{brukner+} if and only if $C$ is positive and  $\Tr  [  C  D]   =  1$
 for every  operator $D$ satisfying the conditions 
 \begin{align}\label{nosiga}
 \Tr_{A_1^{\rm out}}  [D]   =  I_{A_1^{\rm in}}  \otimes    \widetilde B \, ,   \qquad  \Tr_{B_1^{\rm out}}  [\widetilde B]   =   I_{B_1^{\rm in}}  
 \end{align}
 and 
 \begin{align}\label{nosigb}
 \Tr_{B_1^{\rm out}}  [D]   =  I_{B_1^{\rm in}}  \otimes    \widetilde A \, ,   \qquad  \Tr_{A_1^{\rm out}}  [\widetilde A]   =   I_{A_1^{\rm in}} \, ,  
 \end{align}
 with suitable operators $\widetilde A$ and $\widetilde B$.  
 \end{prop}
  For the proof,  see Theorem 2 of \cite{chiribella-dariano-2013-pra}.  The operator $D$ represents the Choi operator of a no-signalling channel \cite{beckman-2001-pra,piani-2006-pra,dariano-2011-prl}, that is, a channel that  prevents the transmission of information from Alice to Bob and from Bob to Alice.  The intuitive idea is that whenever a network can be connected with two local channels, it can also be connected with a no-signalling channel.    

In the following we will denote    by $\mathsf  { NoSig }  (A_1^{\rm in} \to A_1^{\rm out} \, |\,  B_1^{\rm in} \to  B_1^{\rm out})$ is the set of  positive operators satisfying the no-signalling conditions   (\ref{nosiga}) and (\ref{nosigb}).   
With this notation, Proposition \ref{prop:N} can be reformulated as 
\begin{cor}  
An operator $C$ is the Choi operator of a non-causal network as  in  Fig.  \ref{brukner+} if and only if 
\begin{align}
\label{pos}C&\ge0\\
\label{dualnosig}C  &\in  \overline  {\mathsf  { NoSig }  (A_1^{\rm in}\to A_1^{\rm out}| B_1^{\rm in}\to  B_1^{\rm out})} \, ,
\end{align}  where $   \overline  {\mathsf  { NoSig }  (A_1^{\rm in}\to A_1^{\rm out}| B_1^{\rm in}\to  B_1^{\rm out})}$ is the dual affine space of the set of no-sinalling channels.    
\end{cor}
We will denote by  $\Dual   {\mathsf  { NoSig }  (A_1^{\rm in}\to A_1^{\rm out}| B_1^{\rm in}\to  B_1^{\rm out})}$ the set of operators satisfying conditions (\ref{pos}) and (\ref{dualnosig}). The set  $\Dual   {\mathsf  { NoSig }}$ is the set containing all the Choi operators of the non-causal networks of actin on pairs of local operatinos. 



\subsection{The max relative entropy of signalling}   
  
  In some situations, such as the study of non-causal games \cite{oreshkov-2012-nature},  it is natural to search for the non-causal networks that maximize a certain figure of merit.     For example,   consider an experiment where Alice and Bob probe a non-causal network as in Fig.  (\ref{brukner+}).   In their local laboratories, Alice and Bob   measure  the  output systems of the network  with the  POVMs $\{   P_i\}_{i=1}^K$ and $\{  Q_j\}_{j=1}^L$, respectively, and  prepare inputs for the systems, say  $\rho$  and $\sigma$, respectively. The outcomes $i$ and $j$ are  assigned a score $\omega (i,j)$, which quantifies the performance of the non-causal network.   For example, Alice and Bob may want to quantify how much the network correlates their outcomes, corresponding to the score $\omega  (i,j)  =  \delta_{ij}$.    
  More generally, Alice and Bob can probe the network by preparing correlated   states, applying local interactions,   and performing local measurements.  

Describing the test with a performance operator $\Omega$, the  maximum score is achievable by quantum non-causal networks is
\begin{align}\label{omegamax1}
\omega_{\max}   &  =  \max_{C  \in  \Dual  {\mathsf  { NoSig }  (A_1^{\rm in}\to A_1^{\rm out}| B_1^{\rm in}\to  B_1^{\rm out})}  }\,   \<  \Omega,  C\>  \, .
\end{align}
Finding the network that achieves maximum score is similar to finding the entangled state that maximizes the violation of a Bell inequality.    The optimization task can be tackled with our Theorem \ref{lem:geomdual}, which  provides a dual expression for the maximum score:  
\begin{prop}
Let   $\Omega \in  \Herm  \left(   \spc H_{A_1^{\rm out}} \otimes \spc H_{A_1^{\rm in}}\otimes \spc H_{B_1^{\rm out}}  \otimes \spc H_{B_1^{\rm in}}\right)$ be a generic performance operator a  $  \omega_{\max}$ be the maximum score defined in Eq. (\ref{omegamax1}).  Then, one has
\[  \omega_{\max}    =     \,   \min_{\Gamma \in    {\mathsf  { NoSig }  (A_1^{\rm in}\to A_1^{\rm out}| B_1^{\rm in}\to  B_1^{\rm out})} }  \min \{  \lambda  \in   \R  ~|~     \, \lambda \Gamma  \ge \Omega  \}\, . \]
When $\Omega$ is positive, the maximum score is given by 
\begin{align}
\omega_{\max}    =      2^{  D_{\max}    \left ( \,        \Omega  \,   \| \,   \mathsf{ NoSig}    (A_1^{\rm in}\to A_1^{\rm out}| B_1^{\rm in}\to  B_1^{\rm out})   \,        \right )}  \, . 
\end{align}
\end{prop} 

In words: the maximum score  achieved by quantum non-causal networks is determined by the deviation of the performance operator from set of (Choi operators of) no-signalling channels.  We call  $D_{\max}    \left ( \,        A  \,   \| \,   \mathsf{ NoSig}    (A_1^{\rm in},A_1^{\rm out}| B_1^{\rm in},  B_1^{\rm out})   \, \right)$  the \emph{max relative entropy of signalling}, in analogy with the   \emph{relative entropy of entanglement} of a state $\rho$  \cite{bennett-bernstein-1996-pra,popescu-1997-pra,vedral-1997-prl}.

\subsection{Optimizing multipartite non-causal networks}      

The results presented in the bipartite case  can be easily generalized to multipartite non-causal networks.    Consider a quantum network  that can interact with  $k$ local devices, by providing an input system to each device and annihilating its output system.       As in the bipartite case, the network  can be represented by its Choi operator $C$, which will have to satisfy the condition  
\[    \Tr  \left[  C    (   \widetilde A_1  \otimes     \widetilde A_2  \otimes   \cdots \otimes   \widetilde A_k    ) \right ]  =1  \, , \]  
for every set of Choi operators $( \widetilde A_1 \,,        \widetilde A_2 \, , \cdots,    \widetilde A_k)$ representing local quantum channels.   
Equivalently,  the normalization condition can be expressed as  
\[  \Tr  [ C  \,  D]   =  1  \, ,\]
for every Choi operator $D$ representing a  $k$-partite no-signalling channel.  Specifically, the set of Choi operators representing $k$-partite no-signalling channels is defined as follows: 
\begin{defi}
  An operator $D$, acting on  $\bigotimes_{i=1}^k   \left (  \spc H_{A_i^{\rm out}}\otimes \spc H_{A_i^{\rm in}}  \right)  $,  is the Choi operator of a no-signalling channel iff for every subset $  \set J   \subseteq    \{1,\dots, k\} $ one has 
\[    \Tr_{A_{\set J}^{\rm out}}   [    D ]   =  I_{A_{\set J}^{\rm in}}   \otimes    D^c_{  \set J}  \, ,\] 
where $ \Tr_{A_{\set J}^{\rm out}}$ is the partial trace over the Hilbert space $\spc H_{ A_{\set J}^{\rm out} } :  = \bigotimes_{i\in  \set J}  \spc H_{A_i^{\rm out}}$,   $I_{A_{\set J}^{\rm in}}$ is the identity operator on the Hilbert space $\spc H_{ A_{\set J}^{\rm in} } :  = \bigotimes_{i\in  \set J}  \spc H_{A_i^{\rm in} }$, and $D^c_{  \set J} $ is the Choi operator of a quantum channel transforming density matrices on  $\spc H_{ A^{c \, ,in}_{\set J} } :  = \bigotimes_{i  \not \in  \set J}  \spc H_{A_i ^{\rm in}}$ into density matrices on  $\spc H_{ A^{c  \, , out  }_{\set J} } :  = \bigotimes_{i  \not \in  \set J}  \spc H_{A_i^{\rm out} }$.
\end{defi}

We denote the set of $k$-partite no-signalling channels as  $\mathsf {NoSig}_k$, keeping   implicit the specification of the Hilbert spaces. 

Like in the bipartite case, it is natural to consider tasks where one has to find the non-causal network that maximizes a score  of the form $\omega   =  \Tr [    \Omega C   ] $ for some performance operator $\Omega$.     The maximum score is then given by
\begin{align}
\omega_{\max}   =   \max  \left\{  \,   \Tr[  \Omega C ]    ~|~      C  \in    \overline  {\mathsf {NoSig}_k  }   \,  \right\} \, .
\end{align}  

In general,  characterizing the dual affine space of the set of no signalling channels is a rather laborious task.   Using Theorem \ref{lem:geomdual} we can circumvent the problem and express the maximum score as 
\begin{align}
\omega_{\max}   =    \min_{D\in  \mathsf{NoSig}_k}   \min  \left\{  \,   \lambda  \in \R     ~|~    \lambda\, D  \ge \Omega            \,  \right\}  \,      \, ,
\end{align}  
 or, when $\Omega$ is positive 
 \begin{align}   \omega_{\max}    =   2^{    D_{\max}   \left (  \,    \Omega  \, \|  \,    \mathsf {NoSig}_k       \, \right)   }  \, .
 \end{align}   
Again, the performance is determined by the deviation of the performance operator from the set of (Choi operators of) no-signalling channels.

\section{The max relative entropy of non-causal networks}\label{sec:maxnoncausal}  
Like in the case of causal networks, the max relative entropy between two  quantum networks can be related to the max relative entropy of their output states: 
\begin{prop}\label{prop:maxrelativenoncausal}
Let $C^{(0)}$ and $C^{(1)}$ be the Choi operators of two non-causal networks in $\overline {\mathsf {NoSig}_k}$  
and let $E$ be the Choi operator of a network of the form
 \begin{align}\label{scorer2}
 \Qcircuit @C=1em @R=0.2em @!R
{     & \qw \poloFantasmaCn{A_{1}^{\rm in}}  &  \multigate{1}{\map E_1}   &  \qw \poloFantasmaCn{A^{\rm out}_1}    &\qw &    & \qw\poloFantasmaCn{A_2^{\rm in}}  &  \multigate{1}{\map E_2}   & \qw \poloFantasmaCn{A_2^{\rm out}}   &  \qw  &\dots  &&\qw   \poloFantasmaCn{A_N^{\rm in}} &  \multigate{1}{\map E_N}  &  \qw  \poloFantasmaCn{A_N^{\rm out}}   &\qw \\
     & & \pureghost{\map E_1}  &   \qw &  \qw &\qw &\qw & \ghost{\map E_2}  &\qw  &  \qw &  \dots   &&  \qw & \ghost{\map E_N}    & \qw \poloFantasmaCn{S}   &\qw    }\,  ,
\end{align} 
where $S$ is a generic quantum system and the reduced channel     
\begin{align}\label{scorer2}
 \Qcircuit @C=1em @R=0.2em @!R
{     & \qw \poloFantasmaCn{A_{1}^{\rm in}}  &  \multigate{1}{\map E_1}   &  \qw \poloFantasmaCn{A^{\rm out}_1}    &\qw &    & \qw\poloFantasmaCn{A_2^{\rm in}}  &  \multigate{1}{\map E_2}   & \qw \poloFantasmaCn{A_2^{\rm out}}   &  \qw  &\dots  &&\qw   \poloFantasmaCn{A_N^{\rm in}} &  \multigate{1}{\map E_N}  &  \qw  \poloFantasmaCn{A_N^{\rm out}}   &\qw \\
     & & \pureghost{\map E_1}  &   \qw &  \qw &\qw &\qw & \ghost{\map E_2}  &\qw  &  \qw &  \dots   &&  \qw & \ghost{\map E_N}    & \qw \poloFantasmaCn{S}   &\measureD{\Tr}   }\,  ,
\end{align}   is no-signalling.    
Then, one has  
\begin{align}\label{dmaxstates}
D_{\max}   (  C^{(0)} \, \| \,  C^{(1)})   =  \max_E   \, D_{\max}   (     C^{(0)}  *E \, \| \,  C^{(1)}  * E) \, . 
\end{align}
where the maximum runs over all networks of the form (\ref{scorer2}), with arbitrary system $S$. 
\end{prop}  
The proof is the same as the proof of Proposition \ref{prop:maxrelativecausal}.  
The above result shows that the max relative entropy between two non-causal networks is equal to the max relative entropy between the output states generated by connecting the networks to the ``no-signalling part" of a quantum channel, as in figure \ref{fig:noncausalhypothesis}.      

\begin{figure}\label{fig:noncausalhypothesis}
\begin{center}
\hspace*{0cm}\includegraphics[scale=0.5]{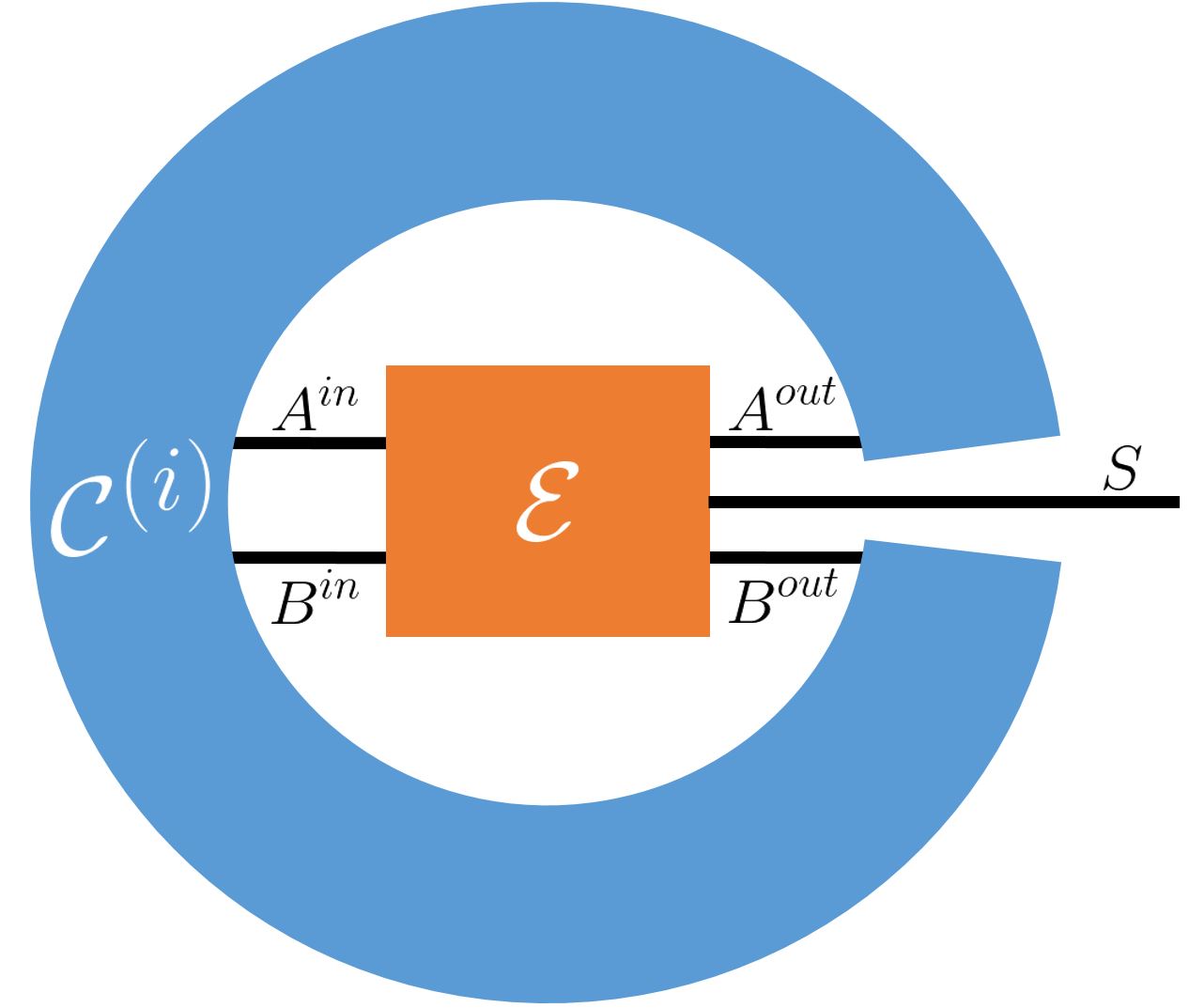}
\end{center}
\caption{Schematic of a test for probing two different hypotheses of quantum spacetime. The two hypotheses are  described by the (possibly non-causal) network (in blue) connecting  systems $A^{\rm in}$ and $ A_1^{\rm out}$ in Alice's laboratory with systems $B^{\rm in}$ and $B^{\rm out}$ in Bob's laboratory.     The test consists in applying a quantum channel $\map E$ (in orange),  acting on systems $A^{\rm in},  A_1^{\rm out} , B^{\rm in}$, and $B^{\rm out}$, plus an additional system $S$.   The channel $\map E$ has the property that, once system $S$ is discarded, the evolution of the remaining systems is no-signalling.   We stress that the above model represents the most general way---\emph{in principle}---to discriminate between two hypotheses of causal structure. However, depending on the situation, there may be constraints on the channel $\map E$, such as the ability to implement $\map E$ with local interactions, or the presence of conservation laws that further limit the set of available channels.   } 
\label{prop}
\end{figure}

Like in the causal case, there is a nice connection to one-shot hypothesis testing. Here one can consider  the problem of distinguishing between two alternative models of spacetime, resulting into different ways to connect the operations performed in $N$ local laboratories.   For example, $C^{(0)}$ could describe a null hypothesis of space time where  all the events  are causally ordered, while $C^{(1)}$ could describe an exotic, non-causal space time.  
 Proposition \ref{prop:maxrelativenoncausal} tells us that, in terms of max relative entropy, the distinguishability of two models of spacetime is   quantified by the max relative entropy of the corresponding non-causal networks.

\section{Applications \label{applications}}
In the following we apply our results to   four optimization problems  involving quantum networks.  
  We will start from the causal case, considering networks that  approximately transform a given set of input channels into a target set of output channels.  Then, we will move the case of non-causal networks.

\subsection{ Transforming quantum channels }

Consider the following scenario: A  black box  implements  a quantum channel in the set $\{ \map E_x\}_{x\in\set X}$, where $\set X$ is an arbitrary index  set. 
  The task is  to simulate another channel  $\map F_x$ using the channel $\map E_x$     as a subroutine.   For example, the black box could implement a unitary gate $U_x$ and the task could be to build the control-unitary gate \cite{friis2014implementing,araujo-2014-njp,nakayama2015quantum,bisio2015quantum}.  
   \[ {\tt ctrl-}U_x  = I  \otimes  |0\>\<  0| +   U_x \otimes  |1\>\<  1| \, .  \] 
To simulate the desired channel  $\map F_x$,  we insert the input channel $\map E_x$ into a quantum causal network, as in the following diagram 
\begin{align}\label{channeltransf}
\Qcircuit @C=1em @R=0.2em @!R
{    &\qw \poloFantasmaCn{A_0}   &    \multigate{1}{\map C_1}    &     \qw \poloFantasmaCn{A_1}    &\gate{    \map E_x } &\qw \poloFantasmaCn{A_2}  &  \multigate{1}{\map C_2}  &  \qw \poloFantasmaCn{A_3}  &\qw      &  \quad   &  =    &\qquad   &\qw \poloFantasmaCn{A_0}   &\gate{    \map E'_x } &\qw \poloFantasmaCn{A_3}  &\qw     \\
  &   &    \pureghost{\map C_1}    &     \qw  &\qw   &\qw  &  \ghost{\map C_2}  &   &   &&&&&& }  
      \quad ,
\end{align}
where $\map C_1$ and $\map C_2$ are suitable quantum channels.
The Choi operator of the output channel $\map E_x'$ is then given by 
\begin{align}  E_x'     =     C  *     E_x   \, ,
\end{align}
where $C$ is the Choi operator of the network and $*$ denotes the link product.  

Let us focus on the case where the target channel $\map F_x$ is an isometry, namely $\map F_x    =  V_x  \cdot V_x^\dag$, with $V_x^\dag V_x  = I$.    
 To measure how close the channel $\map E_x'$ is to the target, we use the  \emph{channel fidelity}  \cite{raginsky-2001-pla,nielsen-2002-pla,horodecki-1999-pra},  given by  
 \begin{align}
 F(\map  E_x',\map F_x):= \frac{1}{d_0^2}  \bb V_x |E_x' |V_x \kk   \, ,  
 \end{align}
 where $d_0$ is the dimension of the input system  $A_0$  and the notation  $|V\kk$ denotes the unnormalized state 
 \[    |V\kk   :=    (  V \otimes   I)  \,    |I\kk  \,  , \qquad |I\kk: =    \sum_{n=1}^d  \,   |n\>|n\>  \, .  \]
In this case, the fidelity can be interpreted as the probability  that the network  passes a test, where the channel $\map E_x'$ is applied locally on one part of an entangled state and  the output is tested  with a POVM containing the  projector on the entangled state $  |V \kk/\sqrt {d_0}$.   The fidelity can be expressed as 
\[     F(\map E_x',\map F_x):= \frac{1}{d_0^2}     \Tr \left[    C     \,  \left( \,   |V_x \kk  \bb V_x |    \otimes E_x^T   \,  \right)   \right]    \]

 Now, if the input channel $\map E_x$ is given with  prior probability $p (x) $, the average channel fidelity is given by  
\begin{align}  
\nonumber F    &  =  \sum_x  \, \ p(x)  \,   F (\map E_x', \map F_x) \\
\label{fidelitychannel}
&  =  \Tr  [ \Omega  \, C   ]   \, ,  \qquad   \Omega  :=  \sum_x \,  p(x)  \,      \left( \,   |V_x \kk  \bb V_x |    \otimes E_x^T   \,  \right)    \ .
 \end{align}

Thanks to Theorem (\ref{dual2}), the maximum fidelity can be expressed as   
\begin{align}\label{maximumfidelity}  F_{\max}   =     \min_{\Gamma \in  \Dual\Comb}     \left\{   \lambda \in\R ~|~    \lambda \, \Gamma \ge \Omega     \right\}   \, ,
\end{align}
where  $\set {DualComb}$ is the set of positive operators on    $\spc H_3 \otimes \spc H_2\otimes \spc H_1\otimes \spc H_0$ satisfying the conditions
\begin{align}  
\nonumber \Gamma   &=    I_3\otimes T_{210} \\
\nonumber    \Tr_2  [  T_{210}] & = I_1  \otimes T_0\\
   \label{setgamma}  \Tr[T_0]  &=  1       \, .
\end{align}

 In the following we illustrate the  use  of this expression in a few examples.  
\subsection{Optimal inversion of an unknown unitary dynamics  \label{gateinversion}}  
Unitary quantum dynamics is, by definition, invertible: given a classical description of a unitary gate $U$, in principle one can always engineer the gate $U^\dag$ implementing the inverse physical process.  However, the situation is different when the gate $U$ is unknown.  Can we devise a   \emph{physical} inversion mechanism, which transforms every  unknown unitary dynamics $U$, given as a black box, into its inverse?   Classically, the analogue of inverting a unitary dynamics is   inverting a permutation. Inverting a permutation with a single evaluation is clearly impossible, because  evaluating the permutation allows us to know its action on one input at most,  and this information is not sufficient to perform an inversion  on the  other inputs.  In the quantum domain, the situation is more interesting, because one use of a unitary gate is enough to store it  faithfully into a quantum memory, by applying $U$ on one side of  a maximally entangled state. A first question is whether the information stored in the memory can be extracted and used to implement the inverse gate $U^\dag$.   Interestingly, this possibility is barred by Nielsen's and Chuang's  no-programming theorem \cite{nielsen-chuang-1997-prl}, which states that only orthogonal states can be used to program unitary gates deterministically and without error.  As an alternative, one can try to think of protocols that simulate $U^\dag$ with one use of $U$, without storing $U$ in a quantum memory.  Protocols of this form are implemented by quantum networks as in Eq. (\ref{channeltransf}).    We now show that even such protocols cannot implement a perfect inversion.     More specifically, we  show that the best way to generate the inverse of an unknown dynamics is simply to  estimate it and to use the  estimate  to implement an  approximate inversion. Our result highlights an analogy between the optimal inversion of an unknown unitary dynamics and the optimal universal NOT (UNOT) gate \cite{buvzek1999optimal,rungta-2001-pra}, the quantum channel that attempts to transform every pure quantum state into its orthogonal complement. A known fact is that no quantum channel can approximate the ideal  UNOT    gate better than a channel that measures the input state and produces an orthogonal state based on the measurement outcome \cite{buvzek1999optimal,rungta-2001-pra}.     Considering this feature, one can think of the unitary inversion as the analogue of the UNOT: they are both involutions and they both are implemented optimally by measure-and-prepare strategies.

Let us assume  that the unknown unitary gate $U$ is drawn at random according to the normalized Haar measure  $\d U$. Then, the performance operator in Eq.   (\ref{fidelitychannel})   takes the form
\begin{align}
\Omega=&\frac{1}{d^2}\int \d U  \, |U^{\dagger}    \kk \bb  U^{\dagger}|_{30} \otimes | \overline{U} \kk\bb  \overline{U} |_{21} \ , 
\end{align}
with $d=   d_0  =  d_1  =  d_2  =d_3$. 
The evaluation of the fidelity is provided in appendix \ref{app:inversion}, where we obtain the value 
\begin{align}
F_{\max}  = \frac 2{d^2}  \,.
\end{align}  

Now, it turns out that the maximum fidelity can be achieved through the estimation of the gate $U$.  Indeed, the optimal strategy for gate estimation is to prepare a maximally entangled state, to apply the unknown gate $U$ on one side, and to perform the POVM  $P_{\widehat U}   =   d \,    | \widehat U  \kk\bb  \widehat U  | $  \cite{chiribella-dariano-2005-pra}. 
This strategy leads to the conditional probability distribution
\[  p(\widehat U  |  U)    =     \left|  \Tr  [U^\dag \widehat U]\right|^2  \, , \]
normalized with respect to the Haar measure.  Averaging the channel fidelity  $F(\hat U, U)  =    \left|\Tr  [U^\dag \widehat U]\right|^2/d^2 $,  we then obtain the value  
\begin{align}
\nonumber F_{\rm est}  (U) &  =      \int  \d \widehat U \,     \left|\Tr  [U^\dag \widehat U]\right|^4/d^2   \\
 \nonumber  &  = 2/d^2  \\
\label{estimation}  &\equiv F_{\max}  \, , \qquad \forall U\in  \set {SU} (d) \, .
  \end{align}
The continuous POVM  with $P_{\hat U}   =  d |\hat U\kk\bb \hat U|$ can also be replaced by a discrete Bell measurement, with $d^2$ outcomes, without affecting the fidelity in the worst case scenario, or equivalently, the average fidelity over all unitaries.   One way or another,  the above discussion proves that    no quantum network can invert a gate better than a classical network that generates the inverse by using gate estimation as  an intermediate step.

\subsection{Simulating the evolution of a charge conjugate particle}
In quantum mechanics,  complex conjugation implements the symmetry   between particles and antiparticles. If the evolution of a quantum particle is described by the  unitary transformation $U$, then the evolution of the corresponding antiparticle  will be described by the unitary transformation $\overline U$,  where each matrix element is replaced by its complex conjugate.  Consider the scenario where one is given a black box that performs a unitary transformation  on a certain particle.     Can we use this black box to simulate the  evolution of the corresponding antiparticle?  Physically,  the most general simulation  strategy is  described by a quantum network as in Eq. (\ref{channeltransf}). 
 
For the charge conjugation problem, the performance  operator $\Omega$ reads
\begin{align}
\nonumber \Omega=&\frac{1}{d^2}\int \d U \,  | \overline{U}  \kk \bb  \overline{U} |_{30} \otimes | \overline{U}  \kk \bb  \overline{U} |_{21}\\
\label{omegaconj} &\frac{1}{d^2} \left( \frac{P_{+,32}\otimes P_{+,10}}{d_+} + \frac{P_{-,32}\otimes P_{-,10}}{d_-} \right) \  \, ,\end{align}
where $P_+$ and $P_-$  ($d_+$ and $d_-$) are the projectors on  (the dimensions of) the symmetric and antisymmetric subspaces,  respectively.  In appendix \ref{app:conj} we evaluate the dual expression Eq. (\ref{maximumfidelity}),  obtaining  the maximum fidelity  
\[F_{\max}  = \frac{2}{d (d-1)} \, .\]   Note that the fidelity is equal to 1 in the case of two-dimensional quantum systems. This  is  consistent with the fact that, for $d=2$, the matrices $U$ and $\overline U$ are unitarily equivalent---specifically, $\overline U  =  Y U Y$, where $Y$ is the Pauli matrix $Y :  = \left(  \begin{array}{cc}   0 &  -i  \\  i  & 0\end{array} \right)$. 
Therefore, one can implement the complex conjugation by sandwiching the original unitary between two Pauli gates.  

For systems of large dimension, the fidelity converges to $2/d^2$, the value achieved by gate estimation [cf. Eq. (\ref{estimation}) in the previous paragraph].    This means that gate estimation is asymptotically the optimal strategy, but, remarkably,   it is not  the optimal strategy when $d$ is finite.  The optimal simulation of the charge conjugate dynamics is achieved by  the network with Choi operator 
\[  C  =   \frac{d\, P_{-,32}}{d_-}\otimes \frac{d\, P_{-,10}}{ d_-}   \, . \]
It is immediate to verify that, indeed, the operator $C$  satisfies the  normalization constraints  and that one has $\Tr [\Omega  C]  =   1/d_-=  F_{\max}$. 
Physically, $C$ represents a ``disconnected network" of the form
\begin{align*}
 \Qcircuit @C=1em @R=0.2em @!R{
 &\qw \poloFantasmaCn{A_0}   &    \gate{\map K}    &     \qw \poloFantasmaCn{A_1}    & \qw &   &\qw \poloFantasmaCn{A_2}  &  \gate{\map K}  &  \qw \poloFantasmaCn{A_3}  &\qw      }  
\quad ,
\end{align*}
 consisting of two subsequent uses of the channel $\map K$ with Choi operator 
 $K   :=   d \, P_-/d_-$.
 When the input gate $U$ is inserted in the open slot,  the overall  evolution from system $A_0$ to system $A_3$ is given by the channel 
 $\map F'   =     \map K \, \map U \,  \map K$, 
which optimally simulates the charge conjugate evolution $\overline U$.   

It is interesting to further elaborate on  the physical meaning  of the operations in the network.   At first, one may guess that the optimal way to conjugate an unknown unitary $U$ is to  approximate the sequence  of transformations  
\begin{align}\label{seq1}  
\rho  \quad  \xrightarrow{~{\rm transpose}~}\quad    \rho^T   \quad   \xrightarrow{\quad U \quad }   \quad   U \rho^T  U^\dag \quad  \xrightarrow{~{\rm transpose}~}     \quad   \overline  U \rho  U^T   \, .
\end{align}    
 As the transpose is not a physical operation, one may try to use the optimal transpose channel \cite{horodecki2002method,horodecki2003limits,buscemi2003optimal,kalev2013optimal,lim2011experimental}, which has Choi operator $T   =  d  P_+/d_+$.  However, this choice would be suboptimal, leading to the fidelity 
 \[F_{\rm transpose}  =   1/d_+  =   2/[d(d+1)]   < F_{\max}  \, . \]       Instead, the optimal strategy is to approximate the \emph{transpose $\tt NOT$},~i.~e.~the impossible transformation that maps every projector  into its orthogonal complement.    In the Heisenberg picture, the transpose $\tt NOT$ maps every observable $A$ into the observable $I-A^T$, allowing us to reproduce the charge conjugate dynamics as 
\[   A   \quad  \xrightarrow{~{\rm transpose~} \tt NOT~}\quad       I  -   A^T      \quad   \xrightarrow{\quad U \quad }   \quad             I  -    U^\dag A^T   U           \quad  \xrightarrow{~{\rm transpose}~\tt NOT~}     \quad   U^T  A \overline U   \, . \]
It turns out that the optimal approximation of the transpose $\tt NOT$ is exactly  the channel $\map K$ used in our  network:   
 in summary, the optimal simulation of the charge conjugate dynamics employs  the optimal transpose $\tt NOT$ instead of the optimal transpose. 
   Some intuition to justify this bizarre fact comes from the observation that the optimal transpose can be implemented via state estimation and, therefore, approximating the sequence (\ref{seq1})  would lead to a classical, estimation-based strategy.  Instead, the transpose $\tt NOT$ cannot be achieved via state estimation. For example, the transpose $\tt NOT $ for qubits is a unitary transformation, corresponding to the Pauli matrix $Y$.

\subsection{Optimal controlization of unknown gates \label{cu}}
Given a  unitary gate $U$,  the corresponding control unitary gate is  \[  {\tt ctrl-}U:=I\otimes|0 \> \< 0| + U \otimes |1\> \<1| \, ,\] where $|0\>$ and $|1\>$ are the states of a qubit acting as control system. Controlization is the task of transforming an unknown gate $U$, accessed as a black box, into the corresponding gate ${\tt ctrl-}U$.  

When $U$ is an arbitrary unitary, perfect controlization is impossible,  as it was recently 
 shown  in Refs. \cite{friis2014implementing,araujo-2014-njp}. Like the no-cloning Theorem, this ``no-controlization" result establishes the impossibility of a perfect functionality. But what about approximate controlization? {\em A priori},  nothing forbids that one could engineer an approximate controlization protocol that achieves high-fidelity, almost circumventing the no-go Theorem.   In the following we show that this is not the case. For a completely unknown unitary gate $U$, we show that  not only is perfect controlization impossible, but also that every quantum strategy for controlization will be at most as good as a  classical strategy that measures the control qubit and performs the gate $U$ or the identity depending on the measurement outcome.

For the controlization task, the performance operator $\Omega$ reads
\begin{align}
\Omega=&\frac{1}{2d^2}\int \d  U  \,  |{\tt ctrl-}U\> \!\> \<\!\< {\tt ctrl-}U| \otimes |\overline{U} \>\!\> \<\!\< \overline{U}| \ .
\end{align}
The evaluation of the maximum fidelity, carried out in \ref{app:ctrl}, yields the optimal  fidelity 
\begin{align*} F_{\max}=\frac 12 \,.
\end{align*} By direct inspection, one can check that this is the same fidelity achieved by a network that measures the control qubit in the computational basis $\{  |0\>,  |1\>\}$ and applies the unknown gate $U$ when the outcome is 1. Specifically, such strategy turns the input gate $U$ into the classically-controlled channel   $\map C_U$ defined by 
\[\map  C_U  (\rho  \otimes \sigma) :  =  \< 0|  \sigma|  0\>  \,    \rho   +    \<1|  \sigma|1\>  \,  U \rho U^\dag \, ,     \]
where $\rho$ is an arbitrary state of the system and $\sigma$ is an arbitrary state of the control qubit. It is immediate to check that the fidelity between the classically-controlled channel $\map C_U$ and the control-unitary gate is $1/2$ for every unitary.  
The above argument shows  that no quantum circuit can perform better than a  classical circuit where the  control qubit is decohered by a measurement.

\subsection{Maximization of the payoff in a  non-causal quantum game  \label{bell}} Here we consider the non-causal  game  introduced by Oreshkov, Costa, and Brukner  in  Ref.  \cite{oreshkov-2012-nature}.   The game involves two spatially separated parties, Alice and Bob, and a referee,  who sends inputs to and receives outputs from the players. Specifically,  the referee sends an input bit $a$ to Alice and two input bits $b$ and $b'$ to Bob. Then, the referee  demands  one output bit $x$ from Alice and one output bit $y$ from Bob.   
  The referee assigns a score $\omega  (x,y|a,b,b')$, given by  
   \begin{align}
 \omega (x,y|a,b,0)= \begin{cases}
  1  &  \qquad     x=b \\
  0 &  \qquad x\not  =  b \, ,
  \end{cases} 
    \qquad &{\rm and} \qquad
 \omega(x,y|a,b,1)= \begin{cases}
  1  &  \qquad  y=a \\
  0 &\qquad y\not =  a \, .
\end{cases}
\end{align}
In this game, Alice and Bob are not subject to the no-signalling constraint. In principle, Alice may be able to communicate to Bob, or  vice-versa.  
The only constraint is that Alice and Bob can interact only through a fixed network, which allows for communication at most in one-way: either from Alice to Bob, or from Bob to Alice.   

It is interesting to see how quantum resources can help Alice and Bob.      The most general quantum resource is described by a network that connects Alice's operations to Bob's operations.    The network will provide inputs $A^{\rm in}$ and $B^{\rm in}$ to Alice and Bob, respectively.   Alice and Bob then perform local operations, transforming systems $A^{\rm in}$ and $B^{\rm in}$ them into systems $A^{\rm out}$ and $B^{\rm out}$.    The  local operations depend on the inputs $a$ and $(b,b')$ and will generate the outputs $x$ and $y$, respectively.  Diagrammatically, this scenario is depicted in  Fig. \ref{game}. 

\begin{figure}
\begin{center}
\hspace*{0cm}\includegraphics[scale=0.5]{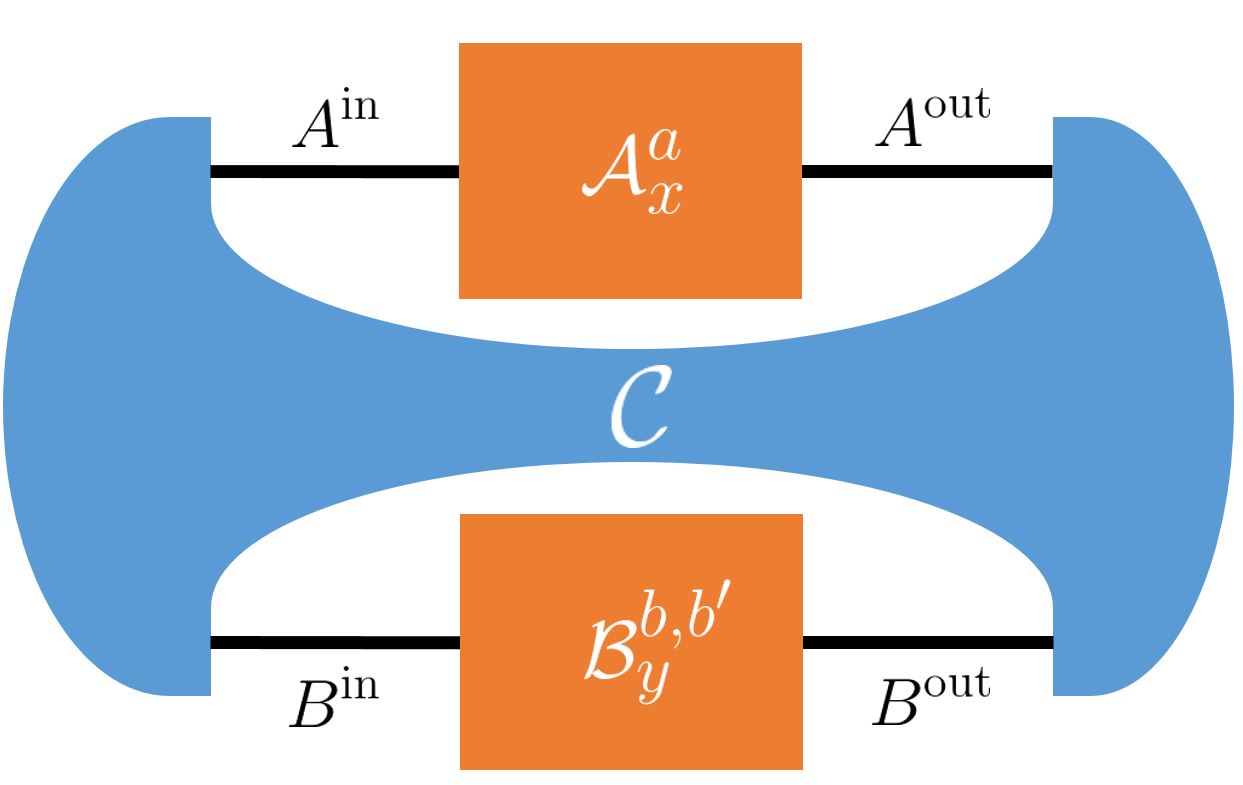}
\end{center}
\caption{The quantum operations $\map A^{a}_x$ and $\map B^{b,b'}_y$ in Alice's and Bob's laboratory interact through a non-causal network $\map C$. The operations act on the Hilbert spaces $\spc H_{A^{\rm in},A^{\rm out}}$ and $\spc H_{B^{\rm in},B^{\rm out}}$ respectively. The  network creates the input systems $A^{\rm in}$ and $B^{\rm in}$   and annihilates the  output systems $A^{\rm out}$ and $B^{\rm out}$.}
\label{game}
\end{figure} 
Mathematically, the  operations are described by two quantum instruments $\left\{  \map M^a_x\right\}_{x=0,1}$  and  $\left\{  \map N^{b,b'}_y\right\}_{y=0,1}$.  With these settings, the probability distribution of the outputs is  given by
\[p(x,y|a,b,b')=\Tr \left[ \left (  M_x^a \otimes N_y^{b,b'}  \right)   \,C  \right ] \, \] 
 where 
  $\{  M_x^a\}_{x=0,1}$ and $\{  N_y^{b,b'}\}_{y=0,1}$ are the Choi operators of Alice's and Bob's instruments, respectively, and $C$ is the Choi operator of the network that mediates the interaction.  
 
With this settings,  the average score is given by 
\begin{align*}
  \omega  &=\frac 18  \,  \sum_{a,b,b',x,y} \,  \omega(x,y|a,b,b')  \, p(x,y|a,b,b')\\
    &  =  \Tr \left[  \Omega  \,  C \right]   \, ,
\end{align*}
where $\Omega$ is the performance operator
\begin{align}\label{noncausalpay}
\Omega    & :=\frac{1}{8}  \,  \sum_{a,b,b',x,y} \,  \omega(x,y|a,b,b')  \,  \left(  M_x^a \otimes N_y^{b,b'} \right) \, .
\end{align}  
   
The main result by  Oreshkov, Costa, and Brukner is that the average score is upper bounded as $\omega  \le 3/4$ whenever the network $C$ has a definite causal order, whereas there exists a non-causal network $C_*$ and local operations $\left\{  \map M^a_{x*}\right\}_{x=0,1}$  and  $\left\{  \map N^{b,b'}_{y*}\right\}_{y=0,1}$  that achieve score
\begin{align}\label{ocb}\omega_*  =  \frac 12 \left(  1  +   \frac1{\sqrt 2}\right) \, .  
\end{align} 
Specifically, the score $\omega_*$ is achieved by choosing systems $A^{\rm in},B^{\rm in},  A^{\rm out}, B^{\rm out}$ to be qubits and by choosing the local operations with Choi operators  
\begin{align} \label{abob}
M_{x*}^a=& \frac{1}{4} \,  \left [I+(-1)^x \sigma_z \right]_{A^{\rm in}}\otimes \left[I+(-1)^a \sigma_z \right]_{A^{\rm out}} \nonumber \\
\nonumber N_{y*}^{b,b'}=&  \,  b' \, \left [   \frac{1}{2}  \left [I+(-1)^y \sigma_z \right]_{B^{\rm in}} \otimes \rho_{B^{\rm out}} \right\} \\
  & + (b'\oplus 1) \left\{  \frac{1}{4} \,  \left[I+(-1)^y \sigma_x \right]_{B^{\rm in}} \otimes \left [I+(-1)^{b+y} \sigma_z \right]_{B^{\rm out}}  \right\} \, .
\end{align}
where $\oplus$ denotes the addition  modulo 2 and $\rho_{B^{\rm out}}$ is a fixed  quantum state on Bob's output,  which  can be chosen  to be the maximally mixed state without loss of generality.

The score $\omega$ can be regarded  as a measure of the non-causality of the network mediating the interactions between Alice and Bob.   
An interesting question is whether $\omega_*$ is the maximum score attainable when Alice's and Bob's instruments  (\ref{abob})  are connected by an arbitrary   non-causal network.      This  question has been indirectly  answered  by Brukner \cite{brukner-2015-njp}, who considered a  more general scenario,   wherein Alice's and Bob's local operations are also subject to optimization.  Brukner showed that the payoff  $\omega_*  =   (1+  1/\sqrt 2)/2$ is maximum over all non-causal networks and over  a certain class of two-outcome instruments on Alice's and Bob's side, allowing Alice's and Bob's systems to have generic dimensions.    When Alice's and Bob's operations are fixed to the  qubit operations (\ref{abob}) used  in the original paper \cite{oreshkov-2012-nature}, we now present an alternative (and comparatively shorter) optimality proof for the value    $\omega_*  =   (1+  1/\sqrt 2)/2$.   This result serves as an illustration of the SDP method, which provides here a nice and straightforward solution. 

Inserting Eq. (\ref{abob}) into Eq. (\ref{noncausalpay})    we obtain the performance operator
 \begin{align*}
\Omega   
&  =  \sum_{i,j,k}   \, |i \> \< i|_{A^{\rm in}} \otimes |j \> \< j|_{A^{\rm out}} \otimes \Omega_{ijk} \otimes |k \> \<k|_{B^{\rm out}}  
\end{align*}
where  $\Omega_{ijk}$ are  operators acting on $B_1$ and are defined as 
\begin{align}
&\Omega_{000}=\frac{1}{8}(|+ \> \< +|+ |0 \> \< 0|), \qquad
&\Omega_{001}=\frac{1}{8}(|- \> \< -|+ |0 \> \< 0|),   \nonumber \\
&\Omega_{010}  =\frac{1}{8}(|+ \> \< +|+ |1 \> \< 1|), \qquad
&\Omega_{011}=\frac{1}{8}(|- \> \< -|+ |1 \> \< 1|), \nonumber \\
&\Omega_{100}  =\frac{1}{8}(|- \> \< -|+ |0 \> \< 0|) , \qquad
&\Omega_{101}=\frac{1}{8}(|+ \> \< +|+ |0 \> \< 0|),   \nonumber \\
&\Omega_{110}=\frac{1}{8}(|- \> \< -|+ |1 \> \< 1|), \qquad
&\Omega_{111}=\frac{1}{8}(|+ \> \< +|+ |1 \> \< 1|) . \nonumber
\end{align}

Now, the dual optimization problem is to find the minimum $\lambda$ such that $\lambda \, \Gamma  \ge  \Omega$, for some Choi operator $\Gamma$ representing a no-signalling channel.  The key observation is that all the  $\Omega_{ijk}$  have the same maximum eigenvalue, equal to  $   e_{\max}   =     1/8  (  1+   1/\sqrt 2)$.    As a result, we can satisfy the  dual constraint by setting  $\lambda  =  1/2  (  1+  1/\sqrt 2)$ and $\Gamma   =  I_{A^{\rm in} A^{\rm out} B^{\rm in} B^{\rm out}}/4$.  Note that  $\Gamma$ is  the Choi operator of a no-signalling channel, as it satisfies Eqs. (\ref{nosiga}) and (\ref{nosigb}).  
Hence, we obtain the bound 
\begin{align}
\label{qubitbound}\omega   \le  \frac 12 \left (  1+   \frac 1 {\sqrt 2}\right) \, ,
\end{align}
valid for every non-causal network.    The bound can be achieved, since r.h.s. matches the value in Eq. (\ref{ocb}). 
\\

\section{Conclusions \label{sec:conclu}}

We developed a semidefinite programming method for the optimization of  quantum networks.  The method can be applied to causal networks as well as more general networks with indefinite causal structure.     For a large class of optimization problems, we observed that the maximum performance  can be expressed in terms of a max relative entropy. 
 Building on this fact, we  extended  the notions of conditional min-entropy and max relative entropy from quantum states to quantum networks.   Specifically, the relative entropy between two networks can be characterized  as  the maximum of the relative entropy between the states that can be generated by the two networks.  
Similarly, the min-entropy of a quantum causal network can be characterized as the maximum  min-entropy that the network can build by interacting over time with a sequence of quantum devices.  Intuitively, the network min-entropy can be   regarded as a measure of the amount of quantum correlations generated over a sequence of time steps.


Our results  have applications  to  a number of scenarios, including~e.~g.~the optimization of algorithms for quantum causal discovery \cite{ried2015quantum}, tomography of quantum  channels and causal networks  \cite{scott2008optimizing,bisio2009optimal,chiribella-dariano-2009-pra,ringbauer2015characterizing}, and quantum machine learning \cite{aimeur2006machine,bisio-chiribella-2010-pra,paparo2014quantum}. 
 Another stimulating avenue of future research is on the quantum engineering side, where our method can be adapted to deal with  optimization tasks  in the presence of  limited energy resources.  For example,  it is interesting to explore the causal networks that can be implemented at zero-energy cost, extending to the network scenario the results obtained in Ref. \cite{chiribella2015optimal} for individual state transitions. The interesting aspect here is the possibility to borrow energy resources at a certain time and to return them at later times, resulting in an overall zero energy balance.   As a further step, the extension from quantum networks working in the  zero-energy regime to network using bounded energy resources is even more compelling in view of future applications. 
  Exploring  how  energy and coherence across energy eigenstates  can be optimally allocated within a distributed system is expected to unveil new quantum advantages, leading to a new layer of optimization in the design of quantum technologies.

\medskip
{{\bf Acknowledgments.} 
We acknowledge the referees of this paper for useful comments that helped improving the presentation and strengthening our results. 
The research of our group is supported by the  Foundational Questions Institute (FQXi-RFP3-1325 and FQXi-MGA-1502),   the National Natural Science Foundation of China through Grant  No. 11675136,  the Hong Kong Research Grant Council
through Grant No. 17326616,   the Canadian Institute for Advanced Research,  the HKU Seed Funding for Basic Research, and   the John Templeton Foundation. GC is grateful to F Buscemi and YC Liang for useful discussions and to A Ac\'in, M Hoban, and R Chavez for organizing the workshop ``Quantum Networks",  Barcelona, March 30-April 1 2016, which offered the occasion for a stimulating exchange of ideas that benefitted this paper.}

\bigskip 
\bibliographystyle{unsrt}
\bibliography{combref}


\appendix

\section{Proof of Theorem \ref{lem:geomdual}}\label{app:theo1}

\Proof
By definition, the value of the primal problem is given by 
\begin{align*}
\omega_{\rm primal}   &  =    \sup\{    \<  A,  X\>  ~|~   X\ge 0 \, ,  X  \in  \set S   \, ,\} \\
&    =    \sup  \left\{  \< A,  X  \>  ~|~    X\ge 0  \, , X  \in   \overline {\overline 
 {\set S}}  \,  \right\}    \\
 &  =   \sup  \left\{  \< A,  X  \>  ~|~    X\ge 0  \, ,   \<   \Gamma,  X\>   =1 \, , \forall \Gamma \in \overline 
 {\set S}  \,   \right\}   \, ,
 \end{align*}
 having used the  relation $\set S   =\overline {\overline 
 {\set S}} $.  
   Now, let us pick an affine  basis for $\overline {\set S}$, say $(\Gamma_i)_{i=1}^K $  and re-write the value of the primal problem as 
   \begin{align*}
\omega_{\rm primal} &   =    \sup \{  \< A,  X  \>  ~|~  X\ge 0 \, ,    \<\Gamma_i ,  X\>  =  1 \, ,   \,\forall  i\in\{  1,\dots,  K\}  \,  \} \, .   
\end{align*}
Weak duality then yields the relation 
\begin{align}
\label{internalslater}\omega_{\rm primal}   &\le     \inf \left\{   \sum_{i=1}^K   \lambda_i ~|~   \lambda_i  \in \R  \, ,      \sum_i  \lambda_i  \,  \Gamma_i  \ge A \right\}  \\
\label{nonzero} &\le     \inf \left\{   \sum_{i=1}^K   \lambda_i ~|~   \lambda_i  \in \R  \, ,      \sum_i  \lambda_i  \,  \Gamma_i  \ge A  \, ,   \sum_i  \lambda_i     \not = 0 \right\}  \\ 
\label{dualset} &  = \inf_{\Gamma\in \overline {\set S}}  \, \min \left\{   \lambda \in \R  ~|~  \lambda  \Gamma  \ge  A   \ \right\}  \,  ,    
\end{align}
having defined $\lambda:  =  \sum_i    \, \lambda_i$ and $\Gamma  :  =  \sum_i  \lambda_i \Gamma_i  /\lambda$.  
  
Now, suppose that  $\set S$ contains a positive operator  $X_0$ and $\overline {\set S}$ contains a strictly positive operator  $\Gamma_0$,  then Slater's Theorem implies the equality: indeed, one can choose the  affine basis  $(\Gamma_i)_{i=1}^K $  to contain the operator $\Gamma_0$. Since $\Gamma_0$ is strictly positive, one can find strictly positive coefficients $  (\lambda_i)_{i=1}^{K}$ such that $ \sum_{i}  \, \lambda_i  \,  \Gamma_i  \ge  A$.   This means that the dual problem in the r.h.s. of Eq. (\ref{internalslater}) admits a strictly positive solution.  Hence, proposition     \ref{Slater} implies the equality in Eq. (\ref{internalslater}).  The equality holds also in   Eq. (\ref{nonzero}), because every solution with $\sum_i  \lambda_i  =  0$ can be replaced by a new solution with   $\sum_i  \lambda_i'  =  \epsilon$, by substituting $\lambda_1$ with $\lambda_1  + \epsilon$,  $\epsilon >0$.  Since $\epsilon$ can be arbitrarily small, this substitution does not change the value of the  infimum.  If $A$ is positive, then one has the lower bound $\omega_{\rm primal}    \ge  \<  A  ,   X_0\>  \ge  0 $.   Eq. (\ref{dualset}) then implies that every $\lambda$ satisfying $\lambda \Gamma \ge A$,   $\Gamma \in \overline {\set S}$ must be  non-negative.  If $\lambda$ is strictly positive, the operator $\Gamma$ must be positive.  If $\lambda=  0$, the operator $\Gamma$ can be chosen to be positive without loss of generality.   In conclusion,   the infimum in Eq. (\ref{dualset}) can be restricted to $\overline {\set S}_+$.   Setting $w : =  1/\lambda$ one finally  obtains the desired expression.

\qed
  

\section{Proof of Theorem \ref{dual2}}\label{app:theo2}

\Proof  The maximum performance is given by Eq. (\ref{omegamax}). The expression can be re-written as 
 \begin{align}  \omega_{\max}     &:  =  \max  \left\{    \,           \<    \Omega, C\>~|~  C  \in \set S  \, , C\ge  0  \right\}   \, , 
 \end{align}
where $\set S$ is  the  affine space of all the   operators on   $  \bigotimes_{j=1}^{N} \, \left(   \spc H_j^{\rm out} \otimes \spc H_j^{\rm in}  \right)$ that are Hermitian and satisfy the linear constraint  (\ref{tr}). Note that  $\set S$ contains the strictly positive operator 
\begin{align}\label{C0}
C_0   =  I_0\otimes I_1 \otimes  \cdots \otimes I_{2N-1}/(d_1 d_3  \dots  d_{2N-1})
\end{align}
and  the dual affine space $\overline {\set S}$  contains the strictly positive operator 
\begin{align}\label{gamma0}
\Gamma_0   =  I_0\otimes I_1 \otimes  \cdots \otimes I_{2N-1}/(d_0 d_2  \dots  d_{2N-2}) 
\end{align}
 Since the sets $\set S$ and $\overline {\set S}$ contain strictly positive operators, the expression in Theorem \ref{lem:geomdual} holds with the equality.    Moreover, one can choose the performance operator $\Omega$ to be positive without loss of generality:  if $\Omega$ is not positive, one can define $\Omega'   =   \Omega  +   c  \Gamma_0$, where $c$ is a positive constant and $\Gamma_0$ is the operator in Eq. (\ref{gamma0}). This substitution only shifts the primal and dual values by the constant $c$, while preserving the optimal solutions.    For the shifted problem, Theorem \ref{lem:geomdual} guarantees that  the dual optimization can be restricted to the positive operators in  $\overline {\set S}_+ $, namely  
\[\omega'_{\max}  =    \inf_{  \Gamma \in \overline  {\set S}_+  }  \min\{\lambda \in \R ~|~  \lambda \Gamma \ge \Omega'  \} \, .\] 
Now, the set $\overline {\set S}_+$ has  been characterized in   Ref.   \cite{chiribella-dariano-2009-pra}: precisely, $\overline {\set S}^+$ is the set of all positive operators $\Gamma$ satisfying the linear constraint 
\begin{align}
\Gamma  =& I_{A_N^{\rm out} }\otimes \Gamma^{(N)} \nonumber \\
\Tr_{A_n^{\rm in}}[\Gamma^{(N)}]=&I_{A^{\rm out}_{n-1}}\otimes \Gamma^{(N-1)} \ , \quad n=2,\ldots,N \nonumber \\
\label{trgamma} \Tr_{A_1^{\rm in}}[\Gamma^{(1)}]=&1 \ , 
\end{align} 
for suitable  positive operators $\Gamma^{(n)}$ acting on  $ \spc H_n^{\rm in}  \otimes  \left[    \bigotimes_{j=1}^{n-1}   \left(  \,  \spc H^{\rm out}_j  \otimes \spc H_{j}^{\rm in} \right)\right]$.    Observing that $I^{A^{\rm out}}_N$ is the Choi operator of the trace channel $\Tr_{A_N^{\rm out}}$ and comparing Eq. (\ref{trgamma}) with Eq. (\ref{tr}) we then obtain that every operator $\Gamma$  in   $\overline {\set S}^+$ is the  Choi operator of a network of the form (\ref{dualcircuit}).    Hence, $\overline {\set S}^+   =  \Dual\Comb$.    Finally, note that the set $\Dual\Comb$ is compact and therefore the infimum is a minimum. 
\qed  

\section{Proof of Proposition \ref{prop:maxoutput}}\label{app:maxoutput}

\Proof     By definition, the max relative entropies are  given by  $D_{\max} (C_0 \|  C_1 )    =   -\log \max  \set W$ and $ D_{\max}  \left (  \sqrt \Gamma   C_0   \sqrt \Gamma  \, \| \,     \sqrt \Gamma   C_1   \sqrt \Gamma  \right  )    =   - \log \max \set W(\Gamma) $, 
 with
  \begin{align*}
 \set W  &: =   \{  w  \in  \R  ~|~  w    C_0   \le  C_1     \}\\
 \set W (\Gamma)  &:  =   \left\{  w  \in  \R  ~|~  w        \sqrt \Gamma C_0  \sqrt \Gamma   \le    \sqrt \Gamma  C_1     \sqrt \Gamma \right \} \, .
 \end{align*}
 By construction, one has $\set W   \subseteq  \set W  (\Gamma) $ for every $\Gamma$, and therefore 
\[   D_{\max}(  C_0 \,  \|     \,C_1)      \ge       D_{\max}  \left (  \sqrt \Gamma   C_0   \sqrt \Gamma  \, \| \,     \sqrt \Gamma   C_1   \sqrt \Gamma  \right  )  \, .  \]
  On the other hand, if $\overline{\set S}_+$  contains a full-rank element  $\Gamma_*$, then $\set W  (\Gamma_*)   =   \set W$.  
  \qed

\section{Proof of Proposition \ref{entropymotivation}}\label{app:entropyproof}

\Proof  Let us compute the conditional min-entropy of the output state  
\[\rho  =    \sigma*   D_1  *  E_1  *D_2  *E_2  *  \cdots  *  E_{N-1}  *  D_N   \in  \St \left(  B_N^{\rm out}  \otimes \spc B_N^{\rm out\prime}   \right)  \]
[cf. Eq. (\ref{state})].  
By the operational characterization of the conditional min-entropy [Eq. (\ref{entropystates})], we have  
\begin{align}   H_{\min}   ( B_N^{\rm out}  \,  |   \,  B_N^{\rm out'} )_\rho  &  =     \max_{  \begin{array}{c}  C  \ge 0  \, , \\  \Tr_{B_N^{\rm out}}  [C]  =  I_{B_N^{\rm out '}}    \end{array}}  \frac{  \Tr  [ \, \rho\,  C^T\, ]}{d_{B_N^{\rm out}} }  \, ,  \end{align}
where $C$ is the Choi operator of a recovery channel  $\map C$, which attempts to turn $\rho$ into the maximally entangled state $|\Phi\>$. Substituting the expression for $\rho$ and maximizing over the sequence $(\sigma,  E_1, \dots,  E_{N-1})$, we then obtain   
  \begin{align} 
 \nonumber  & \max_{\sigma, E_1, E_2, \dots , E_{N-1}}      \,     H_{\min}   ( B_N^{\rm out}  \,  | \,  B_N^{\rm out'} )_\rho   \\  
 \nonumber      & \qquad   = \max_{\sigma, E_1, E_2, \dots , E_{N-1}, C}  \,  \frac{\Tr \left[  \left(  \sigma*  D_1 *  E_1  *  D_2 *  E_2  * \dots *  E_{N-1}  * D_N \right)  \,   C^T \,  \right ]}{d_{B_N^{\rm out}}}   \\
     & \qquad   = \max_{\sigma, E_1, E_2, \dots , E_{N-1}, C}  \,  \frac{   \left(  \sigma*  D_1 *  E_1  *  D_2 *  E_2  * \dots *  E_{N-1}  * D_N \right) *    C  }{d_{B_N^{\rm out}}}   \\
      & \qquad   = \max_{\sigma, E_1, E_2, \dots , E_{N-1}, C}  \,  \frac{    E'  *    R }{d_{B_N^{\rm out}}} \, ,
       \end{align}
   having defined  \[   E'    =  \sigma*  E_1\cdots *E_{N-1}   *  C  \qquad {\rm and}  \qquad R  =  D_1 * \cdots  D_N  \, ,  \]  
     Now, note that $E'$ is the Choi operator of a network of the form of Eq. (\ref{scorer}). 
      Moreover, since the channel $\map C$ can be chosen to be the identity, $E'$ is the Choi operator of an \emph{arbitrary} network of the form of Eq. (\ref{scorer}).    Using   Eqs.  (\ref{Fidelity}) and (\ref{Fmax}) we finally obtain   
 \[   \max_{\sigma, E_1, E_2, \dots , E_{N-1}}      \,     H_{\min}   ( B_N^{\rm out}  \,  | \,  B_N^{\rm out'} )_\rho    =   H_{\min}  (t_N|  t_1\dots  t_{N-1})_R  \,.\]
\qed

\section{Proof of Proposition \ref{prop:maxrelativecausal}}\label{app:maxrelativecausal}  

 \Proof    The proof is based on Proposition \ref{prop:maxoutput}.    Take an operator  $\Gamma \in \Dual\Comb \left( \spc H_{A_1}^{\rm in}\, ,  \spc H_{A_1}^{\rm out} ,   \dots \, ,   \spc H_{A_N}^{\rm in}\, ,  \spc H_{A_N}^{\rm out}  \right)$ and diagonalize it as $\Gamma  =   \sum_i \,   g_i\,    |\phi_i\>\<\phi_i|$. 
 Choose $S$ to be the composite system $A_1^{\rm in}A_1^{\rm out}  \cdots A_N^{\rm in}  A_N^{\rm out}$ and define the vector  
 $  |\Psi\>  =   \sum_i \,  \sqrt {g_i} \,    |   \phi_i\>|\overline \phi_i\>  \in  \spc  H_{A_1^{\rm in}}  \otimes \spc  H_{A_1^{\rm out}}  \otimes \cdots   \spc  H_{A_N^{\rm in}}  \otimes \spc  H_{A_N^{\rm out}}  \otimes \spc H_S$.    Then, the positive operator  $E   =  |\Psi\>\<\Psi|$ is the Choi operator of a network of the form  of Eq. (\ref{scorer1}), as one can check from Eq. (\ref{tr}). 
  Then, explicit calculation gives 
\[    C^{(x)}   *    E    =     \sqrt \Gamma  C^{(x)}  \,  \sqrt \Gamma \, .\]   
Using Proposition (\ref{prop:maxoutput}) we then conclude the equality  
\[  D_{\max}   (C^{(0)} \,\| \, C^{(1)}  )      =   \max_\Gamma  \,     D_{\max}   (  C^{(0)}   *    E   \, \| \,  C^{(1)}   *    E)  \, .  \]        \qed  

\medskip

\section{Normalization condition for supermaps on product channels \label{app:condition}}   
Equation  (\ref{localsupermapchoi})  gives us the Choi operator   $C$. In order for $C$ to be the Choi operator  of a channel, we must have 
\begin{align}
\Tr_{ A_2^{\rm out}, B_2^{\rm out} }  [  C]   =         I_{A_2^{\rm in}} \otimes    I_{B_2^{\rm in}} \, .
\end{align}
Inserting Eq. (\ref{localsupermapchoi}), we then obtain the condition
\begin{align}\label{choicondition}
\Tr_{    A_2^{\rm out},  B_2^{\rm out}, A_1^{\rm in},  A_1^{\rm out}, B_1^{\rm in}, B_1^{\rm out} }    \left [     (   I_{A_2^{\rm out}}  \otimes I_{A_2^{\rm in}}  \otimes  N  \otimes  I_{\map B_2^{\rm out}}  \otimes   I_{  \map B_2^{\rm in}})   (  A  \otimes B)  \right]    =    I_{A_2^{\rm in}}  \otimes I_{B_2^{\rm in}} \, ,
\end{align}
which must be satisfied whenever $A$ and $B$ satisfy the conditions  
\begin{align}\label{choiachoib}
\Tr_{A_1^{\rm out},  A_2^{\rm out}}[A]=   I_{A_1^{\rm in}}  \otimes I_{A_2^{\rm in}}   \qquad {\rm and}    \qquad \Tr_{B_1^{\rm out},  B_2^{\rm out}}[B]=   I_{B_1^{\rm in}}  \otimes I_{B_2^{\rm in}}  \, .
\end{align}

Now, we have the following 
\begin{prop}
For every  operator $N$, the following conditions are equivalent:  
\begin{enumerate}
\item $N$ satisfies the condition    (\ref{choicondition}) for  every operators $A  $ and $B$ satisfying the condition  (\ref{choiachoib})
\item $N$ satisfies the condition  
\begin{align}\label{choicondition2}
\Tr_{ A_1^{\rm in},  A_1^{\rm out}, B_1^{\rm in}, B_1^{\rm out}}  [   N  \left(   \widetilde A  \otimes \widetilde B \right ) ]  =1  \end{align}
for every operators $\widetilde A  \in  \Lin  (  \spc H_{A_1^{\rm out}}\otimes   \spc H_{A_1^{\rm in}})$ and $\widetilde B  \in  \Lin  (  \spc H_{B_1^{\rm out}}\otimes   \spc H_{B_1^{\rm in}})$  satisfying the conditions  
\begin{align} \label{choiachoib2}
\Tr_{A_1^{\rm out}}[\widetilde A]=   I_{A_1^{\rm in}}    \qquad {\rm and}    \qquad \Tr_{B_1^{\rm out}}[ \widetilde B]=   I_{B_1^{\rm in}} \, .
\end{align}
\end{enumerate}
\end{prop}
\begin{proof}
Suppose that the operators $\widetilde{A}$ and $\widetilde{B}$ satisfy the trace conditions (\ref{choiachoib2}).  By defining the operators $A$ and $B$ as $A= \widetilde{A}   \otimes I_{A_2^{\rm in},A_2^{\rm out}}/d_{A_2^{\rm in}}$ and $B=\widetilde{B}   \otimes I_{B_2^{\rm in},B_2^{\rm out}}/d_{B_2^{\rm in}}$, we see that  Eq. (\ref{choiachoib}) is satisfied. Then, Eq. (\ref{choicondition}) becomes
\begin{align}
& \Tr_{    A_2^{\rm out},  B_2^{\rm out}, A_1^{\rm in},  A_1^{\rm out}, B_1^{\rm in}, B_1^{\rm out} }    \left [     (   I_{A_2^{\rm out}}  \otimes I_{A_2^{\rm in}}  \otimes  N  \otimes  I_{ B_2^{\rm out}}  \otimes   I_{  B_2^{\rm in}})   (  A  \otimes B)  \right]  \nonumber \\
=& \Tr_{    A_2^{\rm out},  B_2^{\rm out}, A_1^{\rm in},  A_1^{\rm out}, B_1^{\rm in}, B_1^{\rm out} }    \left \{  \frac{ I_{A_2^{\rm out}} }{d_{A_2^{\rm in}}} \otimes  I_{A_2^{\rm in}}    \otimes   \left[ N(\widetilde{A}\otimes \widetilde{B}) \right]  \otimes \frac{  I_{B_2^{\rm out}}}{d_{B_2^{\rm in}}}  \otimes   I_{   B_2^{\rm in}}  \right\}    =    I_{A_2^{\rm in}}  \otimes I_{B_2^{\rm in}}  \, .
\end{align}
The above equation holds if and onyl if   condition  (\ref{choicondition2})  is satisfied.  
Conversely, if the operator $N$ satisfies condition (\ref{choicondition2}) and $\widetilde{A}$ and $\widetilde{B}$ the trace conditions (\ref{choiachoib2}), we obtain
\begin{align}
& \Tr_{    A_2^{\rm out},  B_2^{\rm out}, A_1^{\rm in},  A_1^{\rm out}, B_1^{\rm in}, B_1^{\rm out} }    \left [     (   I_{A_2^{\rm out}}  \otimes I_{A_2^{\rm in}}  \otimes  N  \otimes  I_{ B_2^{\rm out}}  \otimes   I_{ B_2^{\rm in}})   (  A  \otimes B)  \right]  \nonumber \\
=& \Tr_{A_1^{\rm in},A_1^{\rm out},B_1^{\rm in},B_1^{\rm out}}\left[  \left (I_{A_2^{\rm in}}\otimes N \otimes I_{B_2^{\rm in}}\right)\,  \left( \Tr_{A_2^{\rm out},B_2^{\rm out}} [A \otimes   B ] \right) \right] \nonumber \\
=& \Tr_{A_1^{\rm in},A_1^{\rm out},B_1^{\rm in},B_1^{\rm out}}\left[ \left(  I_{A_2^{\rm in}}\otimes N \otimes I_{B_2^{\rm in}} \right)  \,  \left(\overline{A} \otimes \overline{B} \right) \right] 
\end{align}
where we defined $\Tr_{A_2^{\rm out}}[A]=\overline{A}$ and $\Tr_{B_2^{\rm out}}[B]=\overline{B}$. Hence,  Eq. (\ref{choicondition}) holds if and only if 
\begin{align}
\Tr_{A_1^{\rm in},A_1^{\rm out},B_1^{\rm in},B_1^{\rm out}}\left[ \left(  I_{A_2^{\rm in}}\otimes N \otimes I_{B_2^{\rm in}} \right)  \,  \left(\overline{A} \otimes \overline{B} \right) \right]   =  I_{A_2^{\rm in}}  \otimes I_{B_2^{\rm in}} \, .
\end{align}
In turn, the above equation holds if and only if
\begin{align}\label{final}
\Tr_{A_1^{\rm in},A_1^{\rm out},B_1^{\rm in},B_1^{\rm out}}\left[   N   \,  \left(\overline{A}_\rho \otimes \overline{B}_\sigma \right) \right]   = 1   \,   ,  \qquad   \forall  \rho  \in\St  (  \spc H_{A_2^{\rm in}}) \, ,  \forall  \sigma \in  \St (\spc H_{B_2^{\rm in}}) \, ,
\end{align}
where   $\overline{A}_\rho$ and $ \overline{B}_\sigma$ are defined as 
\begin{align}
\overline A_\rho : =    \Tr_{A_2^{\rm in}}  [     (  \rho  \otimes  I_{A_1^{\rm out} A_1^{\rm in}}) \overline A]   \qquad {\rm and}  \qquad \overline B_\sigma : =    \Tr_{B_2^{\rm in}}  [     (  \rho  \otimes  I_{B_1^{\rm out} B_1^{\rm in}}) \overline B] \, . 
\end{align}  
Now, the normalization condition  (\ref{final}) is nothing but Eq.  (\ref{choicondition2}).  The condition is satisfied because  the operators    $\overline{A}_\rho$ and $ \overline{B}_\sigma$   satisfy condition  (\ref{choiachoib2}). 
\end{proof}

\section{Maximum fidelity for the inversion of an unknown dynamics}\label{app:inversion}
The performance operator $\Omega$  reads
\begin{align}
\Omega=&\frac{1}{d^2}\int \d U  \,   |U^{\dagger} \kk \bb U^{\dagger}|_{30} \otimes | \overline{U} \kk \bb \overline{U} |_{21} \nonumber \\
=&\frac{1}{d^2} \int \d U  \,   (I_3\otimes \overline{U}_{0}\otimes \overline{U}_{2}\otimes I_1)  \, \left(  \,   |I  \kk  \bb I  |_{30} \otimes |I  \kk \bb   I|_{21} \, \right)  \, (I_3\otimes \overline{U}_{0}\otimes \overline{U}_{2}\otimes I_1)^{\dagger} \ .
\end{align}
Explicit calculation using Schur's lemma yields the relations 
\begin{align}
 \label{adjoint1} [\Omega,I_3 \otimes U_{2} \otimes I_1 \otimes U_{0}]  &=0  \\
 \label{adjoint2} [\Omega,U_{3} \otimes I_2 \otimes U_{1} \otimes I_0]&=0 \, ,
 \end{align}
 required to hold for every unitary $U$. 
  Explicitly, the operator $\Omega$  is given by 
\begin{align}
\Omega  =   \frac{1}{d^2} \left( \frac{P_{+,31}\otimes P_{+,20}}{d_+} + \frac{P_{-,31}\otimes P_{-,20}}{d_-} \right) \ , 
\end{align}
$P_+$ and $P_-$ are the projectors on the symmetric and antisymmetric subspace, respectively.

 The problem is to find the minimum $\lambda$ such that $\lambda \Gamma  \ge \Omega$,  for $\Gamma$ satisfying the conditions  (\ref{setgamma}).  The first condition requires $\Gamma$ to be of the form $  \Gamma =   I_3\otimes  T_{210}$. Now,   Eq.   (\ref{adjoint2}) implies that, without loss of generality,   the operator $T_{210}$ can be chosen to  satisfy the condition
\begin{align}
 \ [T_{210},I_2 \otimes U_{1} \otimes I_0]& =0 \qquad \forall \ U \in \grp{  SU}  (d)  
 \end{align}
 which in turn implies 
 \begin{align}\label{d5}
 T_{210} = Q_{20}\otimes I_1
 \end{align}
 where $Q_{20}$ is some  positive operator on $\spc H_{20}$.
 Similarly, Eq. (\ref{adjoint1}) implies that we can choose $T_{210}$  to  satisfy the condition
\begin{align}
 [T_{210},U_{2} \otimes I_1 \otimes U_{0} ]=0 \qquad  \forall \ U \in  \grp{SU} (d) \, .
 \end{align}
 Combined with Eq. (\ref{d5}), the above relation implies 
 \begin{align}
 [Q_{20},U_{2}\otimes U_{0}] =0 \qquad  \forall \ U \in  \grp{SU} (d)
 \end{align}
and therefore
\begin{align}
Q_{20}=\alpha P_+ + \beta P_-
\end{align}
 Finally, the last condition in Eq. (\ref{setgamma}) gives  $\Tr[Q_{20}]=1$ and, therefore, 
\[\alpha d_+ + \beta d_- = d\,  \]

The dual constraint $\lambda \, \Gamma \ge \Omega$ then  reads
\begin{align}
  \lambda \left [  \, \alpha (I_{31} \otimes P_{+,20}) +  \beta (I_{31} \otimes P_{-,20} ) \,  \right] \geq \frac{1}{d^2} \left( \frac{P_{+,31}\otimes P_{+,20}}{d_+} + \frac{P_{-,31}\otimes P_{-,20}}{d_-} \right) \ .
\end{align}
Pinching  both sides with the projectors  $P_{+,31}\otimes P_{+,20}$ and $P_{-,31}\otimes P_{-,20}$, one obtains 
\begin{align}
\lambda   \geq \frac{1}{d_+ d^2  \, \alpha}  \qquad {\rm and} \qquad \lambda   \geq \frac{1}{ d^2  \,  ( d-\alpha d_+)    }
\end{align}
By separately considering the cases $d_+ \alpha d^2 \geq  (d-d_+ \alpha) d^2$ and $d_+ \alpha d^2  <  (d-d_+ \alpha) d^2$, we find that the minimum  $\lambda$ is $\lambda_{\min} =2/d^2$.

\section{Maximum fidelity for the charge conjugation of an unknown unitary evolution}\label{app:conj}
The maximization of the fidelity proceeds in the same way as for gate inversion.  
The only difference is that now  the performance operator $\Omega$ is given by Eq. (\ref{omegaconj}), namely 
\begin{align}
\Omega=&\frac{1}{d^2} \left( \frac{P_{+,32}\otimes P_{+,10}}{d_+} + \frac{P_{-,32}\otimes P_{-,10}}{d_-} \right) \ .
\end{align}
The form  of  $\Omega$ implies the relations 
\begin{align} 
\label{sym1} [\Omega,U_{3} \otimes U_{2}  \otimes I_{10}]&=0  \\
\label{sym2} [\Omega, I_{32} \otimes U_{1}\otimes U_{0}]&=0  \, ,
\end{align}
valid for every $U$ in $\grp {SU} (d)$.  Now, one has to find the minimum $\lambda$ such that $\lambda \,   \left(  I_3 \otimes T_{210}  \right)  \ge   \Omega$,  with some $\Gamma $  satisfying Eqs. (\ref{setgamma}).    
  Eq. (\ref{sym1}) implies that, without loss of generality, one has 
  \begin{align}
[T_{210}, U_{2} \otimes  I_{10}]=0     \qquad \forall U\in  \grp{SU}  (d) \, ,
\end{align}
and therefore   $T_{210} =     I_2 \otimes Q_{10}$.  
Moreover, the second condition in Eq. (\ref{setgamma}) reads
\begin{align}
 \Tr_2[T_{210}]=I_1\otimes \rho_0 \nonumber 
 \end{align}
 and implies that $Q_{10}$ has the form $ Q_{10}={I_1\otimes \rho_0}/{d}$. Finally, Eq. (\ref{sym2}) implies that one can choose $\rho_0=  I/d$ without loss of generality. 
 Summing everything up, $\Gamma$ can be chosen to be of the form $\Gamma=I_3\otimes T_{210}=I_{3210}/d^2$. The dual constraint  $\lambda \Gamma \ge \Omega$  then becomes 
  \begin{align}
  \lambda \,   \frac{I_{3210}}{{d^2}} \geq \frac{1}{d^2} \left( \frac{P_{+,32}\otimes P_{+,10}}{d_+} + \frac{P_{-,32}\otimes P_{-,10}}{d_-} \right)
\end{align}
yielding the minimal value  $\lambda_{\min}  =1/d_-  =2/d(d-1)$.
\section{Maximum fidelity for unitary controlization}\label{app:ctrl}
The performance  operator for the controlization problem is 
\begin{align}
\Omega=&\frac{1}{4d^2}\int \d g\,   |{\tt ctrl-}U   \kk \bb  {\tt ctrl-}U|_{30Q'Q} \otimes |\overline{U}   \kk \bb \overline{U}|_{21} \nonumber \\
=&           \Omega_{3210}^{(0)} \otimes     |0\>\<0|_{Q'} \otimes  |0\>\<0|_{Q}     +            \Omega_{3210}^{(1)} \otimes     |1\>\<1|_{Q'}  \otimes |1\>\<1|_{Q}  \, , 
\end{align}
where $Q$ and $Q'$ denote the control qubit before and after the interaction, respectively, and 
\begin{align}
\label{omegauno1} \Omega_{3210}^{(0)}    &:  =   \frac{1}{4d^2} \,    \left(      E_{30} \otimes  I_{21} \right) \\
\label{omegauno2}\Omega_{3210}^{(1)}    &:  =  \frac{1}{4d^2}   \,   \left(   E_{32}   \otimes E_{10}   +\frac{E^\perp_{32} \otimes E^\perp_{10}}{d_\perp}\right) \, . 
\end{align}
Here $E$ denotes  the projector on the maximally entangled state $|\Phi^+\>    =   |I\kk/\sqrt d$, $E_\perp$ is  the orthogonal projector $E_\perp  :=  I^{\otimes 2}  -  E$, and  $d_\perp   :=  d^2-1 $.      Note that the operators   $\Omega_{3210}^{(0)}$ and   $\Omega_{3210}^{(1)}$ satisfies the conditions
\begin{align}
\label{symm0}  \left [\Omega_{3210}^{(1)},  U_{ 3} \otimes  I_{21}  \otimes \overline{U}_{0}  \right]&=0 \\
\label{symm1}  \left [\Omega_{3210}^{(1)},  U_{ 3} \otimes \overline{U}_{2} \otimes I_{10}  \right]&=0 \\
\label{symm2} \left[\Omega_{3210}^{(1)}, I_{32}\otimes \overline{U}_{  1} \otimes U_{ 0}  \right]&=0  \, ,
\end{align}
for every group element $U \in \grp {SU} (d)$. 

To solve the dual problem, we have to find the minimum $\lambda$ satisfying the relation $\lambda\,  \Gamma  \ge  \Omega$ for some dual comb $\Gamma$.   By Eq. (\ref{setgamma}), we have $\Gamma  =    I_b\otimes I_3  \otimes T_{210a}$, for some suitable operator $T_{210a}$ satisfying the conditions  
\begin{align*}
\Tr_2  \left[   T_{210a}  \right]    &=   I_1\otimes   \rho_{0a}   \\
\Tr_{0a}  \left[   \rho_{0a}\right]      &=  1 \, .
\end{align*}  
Without loss of generality, $T_{210a}$ can be chosen of the form 
\begin{align}\label{ta}
T_{210a}   =  T^{(0)}_{210}  \otimes  |0\>\<0|_Q  +      T^{(1)}_{210}  \otimes  |1\>\<1|_Q   \, ,
\end{align}
with  the operators $T^{(0)}_{210}$ and $T^{(1)}_{210}$ satisfying the conditions  
\begin{align}
\label{aa}
\Tr_2  \left[  T^{(0)}_{210}  \right]   =   p_0  \,   \left [  I_1\otimes  \rho_0^{(0)}\right]   \qquad  {\rm and}  \qquad  \Tr_2  \left[  T^{(1)}_{210}  \right]   =   p_1  \,   \left[   I_1\otimes  \rho_0^{(1)}  \right]  
\end{align}
where   $\rho_0^{(0)}$ and $\rho_0^{(1)}$ are  two density matrices and $p_0$ and $p_1$ are probabilities. 
The dual constraint is then reduced to  
\begin{align}\label{dualsplit}
\lambda \, \left[  I_3\otimes  T^{(k)}_{210}  \right]\ge   \Omega_{3210}^{(k)}   \, , \qquad \forall k\in \{0,1\} \, .
\end{align}
  
At this point, Eq. (\ref{symm0}) implies that, without loss of generality, one can choose $T^{(0)}_{210}$ to satisfy the relation  
\[    \left [ T^{(0)}_{210}   ,   I_{21}  \otimes  \overline U_{0} \right]  =  0  \, ,\qquad \forall U\in\grp{SU}  (d)  \, ,\,  \]  
which implies $T^{(0)}_{210}  =  Q^{(0)}_{21}  \otimes  I_0$ for some suitable operator $Q^{(0)}_{21}$.  Moreover, Eq. (\ref{omegauno1}) implies that, without loss of generality, one can choose    $Q_{21}^{(0)}$ to be proportional to the identity, so that, eventually one has
\begin{align}\label{to}
T^{(0)}_{210}   =p_0\,   \frac {  I_2\otimes I_1\otimes I_0}{d^2}   \, . 
\end{align}
Similarly, Eq. (\ref{omegauno2}) implies that, without loss of generality, one can choose   $T^{(1)}_{210}$  to satisfy the relations
\begin{align}
\label{tuno1} \left[T^{(1)}_{210},\overline{U}_{2} \otimes I_{10}\right]  &=0\\
\label{tuno2}   \left[T^{(1)}_{210},I_{2}\otimes \overline{U}_{1} \otimes {U}_{0}\right]&=0 \,, 
\end{align}
for every unitary $U\in\grp {SU}(d)$.    Now, equation (\ref{tuno1}) implies   that $T^{(1)}_{210}$ has the form 
\begin{align}\label{q}  T^{(1)}_{210}=I_{2}\otimes Q^{(1)}_{10} 
\end{align} 
and Eq. (\ref{aa}) implies the condition   
\begin{align*}   d\, Q^{(1)}_{10}    & =        \Tr_{2} \left[T^{(1)}_{210}  \right]  \\
&=  {p_1}    \,      \left  [  I_{1}\otimes \rho^{(1)}_{0}   \right]   
\end{align*}  
for some probability $p_1$ and some quantum state  $\rho^{(1)}_{0}$.   
Combining  Eqs.   (\ref{q}) and  (\ref{tuno2}) one finally obtains   $Q^{(1)}_{10}=p_1 \,   {  I_{1}\otimes  I_{0} }/{d^2}$,
and therefore
\begin{align}\label{ti}
T^{(1)}_{210}   =p_1\,   \frac {  I_2\otimes I_1\otimes I_0}{d^2}   \, .
\end{align}
Inserting the above relations into the dual constraint, we then obtain
\begin{align}
\lambda \,  p_0  \,   \frac{I_{3210}}{d^2}   &\geq \frac{1}{4d^2}   \left( E_{30}  \otimes I_{21}  \right)  \\  
\lambda  \,  p_1 \,  \frac{I_{3210}}{d^2}    & \geq \frac{1}{4d^2} \left( E_{32} \otimes E_{10} +\frac{E^\perp_{32} \otimes E^\perp_{10}}{d_\perp}\right) \, ,\end{align}
having used  Eqs. (\ref{to}), and (\ref{ti}).    To satisfy the constraint, the parameters $\lambda, p_0,$ and $p_1$ must satisfy  $\lambda  p_0  \ge  1/4$ and $\lambda p_1   \ge 1/4$, leading to the minimum value  $  \lambda_{\min}  =     1/2 $.

\end{document}